\documentclass[12pt]{article}
\usepackage{amsfonts}
\usepackage{amsmath}
\usepackage{amsthm}
\usepackage{hyperref}
\usepackage{graphicx}
\usepackage{natbib}
\usepackage{amssymb}
\usepackage[onehalfspacing]{setspace}
\usepackage{subfigure}
\usepackage{booktabs}
\usepackage{multirow}
\usepackage{float}
\usepackage[a4paper]{geometry}
\usepackage{xcolor}

\newtheorem{theorem}{Theorem}

\newtheorem{lemma}{Lemma}

\numberwithin {equation}{section} 
\newtheorem{assumption}{Assumption}

\newcommand{\eigbig}{\overline{\textit{eig}}}
\newcommand{\eigsmall}{\underline{\textit{eig}}}

\theoremstyle{definition}
\newtheorem{example}{Example}
\begin{document}
	
	\title{Testing linearity of spatial
		interaction functions \`a la Ramsey\thanks{This paper is dedicated to the memory of Francesca Rossi. We are grateful to a co-editor and three referees for comments that improved the paper considerably, and to seminar participants at Princeton, Aarhus, Cambridge, Queen Mary, RCEA, Warwick Econometrics Workshop, National University of Singapore, EcoSta, Midwest Econometrics Group 2024, Queen's and Exeter, as well as S\'{i}lvia Gon\c{c}alves and Mikkel S{\o}lvsten, for comments and feedback. Abhimanyu Gupta's research was supported by the Leverhulme Trust via grant RPG-2024-038.}}
	\author{Abhimanyu Gupta \thanks{%
			Department of Economics, Queen's University, Dunning Hall, 94 University Avenue, Kingston, Ontario K7L 3N6, Canada. Email: abhimanyu.g@queensu.ca}  \and Jungyoon Lee \thanks{%
			Department of Economics, Royal Holloway, University of London, Egham, TW20 0EX, UK. Email: jungyoon.lee@rhul.ac.uk} \and Francesca Rossi\thanks{Department of Economics, University of Verona,  via Cantarane 24, 37129, Verona, Italy.}}
	\date{\today}
	\maketitle
	\begin{abstract}
		\noindent	We propose a computationally straightforward test for the linearity of a spatial interaction function. Such functions arise commonly, either as practitioner imposed specifications or due to optimizing behaviour by agents. Our conditional heteroskedasticity robust test is nonparametric, but based on the Lagrange Multiplier principle and reminiscent of the Ramsey RESET approach. This  entails estimation only under the null hypothesis, which yields an easy to estimate linear spatial autoregressive model. Monte Carlo simulations show excellent size control and power. An empirical study with Finnish data illustrates the test's practical usefulness, shedding light on debates on the presence of tax competition among neighbouring municipalities.
		
		\noindent \textit{Keywords:} Spatial autoregression, Series expansion, Specification testing, Nonparametric. 
		
		\noindent  \textit{JEL Classification:} C31, C14.
		
	\end{abstract}
	
	\newpage	
	\section{Introduction}
	\allowdisplaybreaks
	There has been a recent surge in the study of econometric models capturing social interactions between agents, see e.g. \cite{paula_2017} for a survey. By virtue of its close links to network and social interaction models, the spatial econometrics literature has also grown commensurately. A key component of this literature is the linear spatial autoregressive (SAR) model, introduced by \cite{Cliff1968, cliff1973spatial}. The `spatial' moniker is unfortunate in some sense because it belies the generality of the SAR model. In fact the `space' in question can be any economic dimension that connects individuals, e.g. social networks and conflict alliances. In this paper we propose and theoretically justify a convenient nonparametric test for the linearity of the SAR specification, which is by far the most prevalent in the literature.
	
	The assumption of linearity in the preponderance of studies involving SAR models is perhaps unsurprising given that even in straightforward multiple regression the linear model wins the popularity contest. Nevertheless, linearity is a strong assumption, criticized by \cite{Pinkse2010} for instance, and discussed in detail in \cite{paula_2017}, Section 3.2. In the case of spatial autoregressions, nonlinearity in the spatial lag also induces immense technical difficulties due to the simultaneity inherent in the model. This has led to the seminal development of sophisticated machinery \citep{Jenish2009, Jenish2012}, but the flipside is that the conditions tend to be technically challenging and restricted to geographic data.
	
	The econometric implications of nonlinearity in spatial or network interactions are manifold. The focus of this paper is to test for potential nonlinearity in the so called `link' function, which transmits a network connection-weighted average of peer outcomes to an individual's own outcome.  \cite{Tincani2018} finds evidence of nonlinearities in education peer effects in Chile, the link function here operating on the distribution of peer outcomes. Other implications of nonlinearities can include multiple equilibria, or potentially none at all. We refer the reader to detailed discussions of various forms of nonlinearity and their implications in, for example, \cite{Blume2011} and \cite{paula_2017}.
	
	While of importance in its own econometric right, and possibly in a reduced form sense, testing for the linearity of a SAR specification can also address deeper economic considerations. Many network games yield equilibria with linear best responses and, when taken to the data, imply linear SAR specifications. Thus, a rejection of the linearity of the SAR specification can in fact be viewed as a rejection of the underlying structural economic model. Examples of such models include contest success functions (CSFs) in the \cite{Tullock1980} style \citep{konig2017networks}, quadratic utility network games \citep{ballester2006s, calvo2009peer}  and public goods provision games in networks \citep{bramoulle2007public}. These are discussed in more detail in the next section.
	
	Nonlinearity of best responses in network games can have profound economic implications. \cite{acemoglu2015networks} emphasize that convexity or concavity of nonlinear network interactions can shape how the underlying network transmits shocks. Most strikingly, they show that with concave interactions, densely connected economies outperform sparser ones in the sense of expected macroeconomic outcomes. Indeed, an economy in which interlinkages are maximally dense outperforms all other (symmetric) economies. On the contrary, with convex interactions this pattern is completely reversed, with the maximally dense economy being the worst performing. The finding that concavity of interactions acts as a `shock absorber' while convexity turns this into a `shockwave propogator' aligns with the groundbreaking work of \cite{allen2000financial} on financial contagion. Thus, if a theoretical structural model begets a linear best response equilibrium, it behoves the economist to have some statistical confidence in the resulting empirical model.
	
	We employ a residual based test statistic reminiscent of the Lagrange Multiplier (LM) test to avoid nonparametric estimation altogether. The basic idea is to approximate the potentially nonlinear part of the model with a sieve expansion of length $p$. By choosing a basis of polynomial components, a test of linearity may be constructed
	by setting to zero all coefficients except the one corresponding to the linear term. This is an approach in the spirit of the famous Ramsey RESET test described in \cite{Ramsey1969}. By employing the LM principle, we require parameter estimates only under the null hypothesis, which yields the familiar linear SAR model. 
	
	Because $p$ is a smoothing parameter, the number of coefficients on nonlinear terms must grow with sample size. Therefore our test will be for an increasing number of restrictions $p\rightarrow\infty$ asymptotically, and the usual $\chi^2_p$ asymptotic distribution cannot be relied upon. To address this, we will use a centred and scaled
	version of the statistic. In particular our test statistic will be of the form $(\mathcal{T}_p-p)/\sqrt{2p}$, where $\mathcal{T}_p$ is computable by estimating the model under the null of linearity. We will show that this is asymptotically standard normal under the null, noting
	that a $\chi^2_p$ random variable has mean $p$ and variance $2p$, in the spirit of \cite{DeJong1994}, \cite{Hong1995}, \cite{Donald2003} and \cite{Gupta2023}, to cite a few examples. We also establish the test's consistency and ability to detect local alternatives at a suitably dampened nonparametric rate.
	
	There are several specification tests in the literature for linearity of the regression function in a linear SAR model, see for instance \cite{Su2017},  \cite{Gupta2022} and \cite{Chen2025}. On the other hand, the cupboard is rather bare as far as tests of linearity of the spatial lag are concerned. The most direct linearity test is provided by \cite{hoshino2022sieve} and relies on the nonparametric function being estimated. This is a different approach from our Ramsey-style test, which stresses simplicity. The model considered is also somewhat different from ours: we test for linearity of the spatial lag as a whole while \cite{hoshino2022sieve} focuses on testing linearity in the outcome variable. Thus, the two approaches to testing linearity complement each other. Another type of linearity test is proposed by \cite{Malikov2017}, but there the spatial parameter itself is modelled as a varying coefficient while the spatial lag is linear. This is rather different from our setup. Finally, another class of tests is of the omnibus type where model misspecification can arise from a variety of sources, see \cite{Lee2024} and \cite{Yang2024} for tests of the `integrated conditional moments' type. 
	
	In a Monte Carlo simulation study, we show that critical values based on our theory deliver excellent size and power performance for a wide variety of commonly encountered  social interaction networks and spatial links. These include contiguity, nearest neighbours, distance cutoff networks and more general structures. We experiment with both standard normal critical values as well as standardized $(\chi^2_p-p)/\sqrt{2p}$ critical values, and observe that the latter can do well with smaller values of $p$.
	
	The linear SAR model is by far the most commonly employed empirical specification. In an  empirical study, we study a model of tax competition and show how our test for linearity can raise new questions or indeed refine existing analysis. Our test corroborates \cite{Lyytikaeinen2012}'s finding of absence of tax competition between Finnish  municipalities in their property tax setting behaviour, illustrating how  different  model specifications can lead to contrasting conclusions and how our test can offer  insightful guidance on the specification choice.
	
	The next section introduces our basic setup and a number of structural economic models that imply linear SAR specifications. In Section \ref{sec:stat} we define our test statistic based on linear moment conditions, while Section \ref{sec:asymptotics} presents the key assumptions and asymptotic results. In Section \ref{sec:quadmoms} we present an extension to a GMM setup with linear-quadratic moment conditions. Section \ref{sec:MC} contains the results of our Monte Carlo experiments and Section \ref{sec:applications} presents our empirical study of tax competition in Finland. Section \ref{sec:concl} concludes. All proofs are collected in the appendix.
	\section{Basic setup and examples}\label{sec:setup}
	Let $W$ be a spatial weight matrix with rows $w_i^{\prime }$ and zero
	diagonal. Given an $n\times 1$ response vector $y$ and $k\times 1$ covariate
	vector $x_i$ with unity as first element, we commence from  
	\begin{equation}  \label{model}
		y_i= \lambda w_i^\prime y + f\left(w_i^{\prime }y\right)+x_i^{\prime }\beta+\epsilon_i,
		i=1\ldots,n,
	\end{equation}
	with $f(\cdot)$ being an unknown nonlinear function, $\epsilon_i$ independent disturbances and $\beta\neq 0$. The last condition is imposed for now as our first testing procedure relies on linear moments constructed via instrumental variables that deliver identification only if $\beta\neq 0$, see e.g. \cite{kelejian1998generalized}. We permit $\beta=0$ in Section \ref{sec:quadmoms}. The $x_i$ can contain both exogenous and endogenous components. The first term on the RHS represents the usual linear SAR component and $f(\cdot)$ is a nonparametric function that embodies the potential non-linearity of the true data generating process. As a point of notation, note that in SAR models it is common practice to treat $y_i$, $w_i$ and $x_i$ as triangular arrays. Strictly speaking, this would require subscripting with $n$, but for notational simplicity we will omit this for the entire paper.
	
	%	\begin{assumption}\label{ass:basiclambdafW}
	%		The function $f(z): \Re \rightarrow \Re$ has continuous first derivative $f^{(1)}(z) = df(z)/d(z)$ and 
	%		\begin{equation}\label{contraction}
	%			\left (\vert\lambda\vert +	\underset{x}{\sup}\left\vert f^{(1)}(x)\right\vert \right)\underset{n}{\sup}\left\Vert W\right\Vert_{\infty} <1,
	%		\end{equation}
	%		where $\lambda \in \Lambda$ and $\Lambda$ is a suitable parameter space in $\Re$, not necessarily compact.
	%	\end{assumption}  
	
	% Under typical SAR assumptions that impose $\Lambda$ to be a closed subset in $(-1, 1)$ and $\underset{n}{\sup}\left\Vert W\right\Vert_{\infty} <K$, where throughout the paper $K$ denotes a generic arbitrarily large constant that is independent of $n$, (\ref{contraction}) reduces to $\underset{x}{\sup}\left\vert f^{(1)}(x)\right\vert < (1-\vert \lambda\vert)/K$. 
	
	Apart from being an often used econometric specification, the linear SAR model is implied by a number of theoretical models. Thus an empirical test of its linearity in fact serves as test of the underlying economic structure that leads to the estimation of a SAR structure. In this sense our test can be viewed as test of economic theory, rather than simply the statistical fit of an econometric model. We discuss some examples of such theories in this section, with $x_i$ always denoting a $k\times 1 $ vector of observable characteristics for the $i$-th individual, and $\beta$ a $k\times 1$ parameter vector. 
	
	\begin{example}[\textit{Tax Competition}]\label{example:tax}
		Tax competion is the setting of our empirical study in Section \ref{sec:applications}. Indeed, the literature on tax competition justifies the use of linear SAR empirical specifications by extending theoretical results of the \cite{Kanbur1993} model of  tax competition between two revenue-maximising governments in at least two ways.
		
		The first method, see e.g. \cite{Redoano2014} gives rise to linear tax reaction functions in an $n$-country setting.  Households determine their purchase of goods/investment portfolio  over $n$ countries by maximizing their returns net of  transaction costs that are increasing in distances.   The representative household of country $i$, which has capital endowment of $\tilde{k}_i$,  chooses its investment porfolio $k_{ij}$, $j=1, \cdots, n$, to maximize
		\begin{equation*}
			\displaystyle\sum_{j=1}^n k_{ij} ( 1-\tau_j)  - \displaystyle\sum_{j\neq i } c( d_{ij}) \frac{k_{ij}^2}{2}
		\end{equation*}
		subject to $\sum_j k_{ij} = \tilde{k}_i$, where  it is assumed that production  is linear in capital,  $\tau_{j}$ is the tax rate of country $j$, $d_{ij}$ is some notion of  distance between countries $i$ and $j$, and $c(\cdot)$ an increasing cost function. This leads to a quadratic revenue function for governments, so that government $i$ chooses its tax rate $\tau_i$ to maximize its total tax revenue
		\begin{equation*}
			TR_i=\tau_i \left( \tilde{k}_i + \displaystyle\sum_{j \neq i} \frac{\tau_{j}-\tau_i}{c(d_{ij})}   \right).
		\end{equation*} 
		The best responses are then linear in the $\tau_j$ and given by:
		\begin{equation}\label{hbr}
			\tau_i^*=\frac{1}{2 \displaystyle\sum_{j \neq i} \frac{1}{c(d_{ij})}}   \left(\tilde{k}_i + \displaystyle\sum_{j \neq i} \frac{\tau_{j}}{ c(d_{ij})}, \right)
		\end{equation} 
		leading to a linear SAR estimating equation below with a row-normalized inverse-distance weight matrix:
		\begin{equation}\label{taxsar}
			\tau_i=\alpha_i + x_i^{\prime}\beta  + \lambda \sum_{j\neq i} w_{ij}\tau_j + \epsilon_i
		\end{equation}
		where $\alpha_i$ is the country fixed effect.%\footnote{The $1/2$ factor in (\ref{hbr}) is an artifact of the  theoretical model and hence not imposed in estimation, as in \cite{Redoano2014}.}
		
		On the other hand, a second method proposed by \cite{Devereux2007} extends the \cite{Kanbur1993} model to allow individual demand for the taxed good to be price-elastic. While we refer the reader to the original paper for details of the modelling strategy, the empirical model in Section 4 therein employs (\ref{taxsar}) using data on cigarette taxation in US states and a number of economically motivated choices for $w_{ij}$. Yet another theoretical model, this time of corporate tax competition between OECD countries, can be seen to yield (\ref{taxsar}) in \cite{Devereux2008}.
	\end{example}
	
	\begin{example}[\textit{Conflict Networks}]\label{example:conflict}
		\cite{konig2017networks} bring a theoretical and empirical perspective on how a network of military alliances and enmities affects the intensity of a conflict. Their theoretical model combines network theory and politico-economic theory of conflict. Simplifying their setting for ease of exposition, for participants $i$ and $j$ in a conflict define the $n\times n$ adjacency matrix $W^*$ by
		\begin{equation}\label{konigW}
			w^*_{ij}=\begin{cases}
				1, &\text{if } i \text{ and } j \text{ are allies},\\
				0, &\text{if } i \text{ and } j \text{ are in a neutral relationship}.\\
			\end{cases}
		\end{equation}
		The $n$ participants compete for a divisible prize $V$, for example land or resources, and suffer a defeat cost $D\geq 0$ if they do not win a fraction of it. The payoff function maps the groups' relative fighting intensities into shares
		of the prize and can be interpreted as a \cite{Tullock1980}-type contest success function. Specifically, letting $\mathcal{G}^n$ denote the class of graphs on $n$ nodes, the payoff $\pi:\mathcal{G}^n\times \Re^n\rightarrow \Re$ for participant $i$ is 
		\begin{equation}\label{koningpayoff}
			\pi_i(G,y)=\begin{cases}
				\frac{V\varphi_i(G,y)}{\displaystyle\sum_{j=1}^n\max\left\{0,\varphi_j(G,y)\right\}}-y_i, & \text{if } \varphi_i(G,y)\geq 0,\\
				-D, & \text{if } \varphi_i(G,y)< 0,
			\end{cases}
		\end{equation}
		where $y\in\Re^n$ is a vector of fighting efforts and $\varphi_i(\cdot)\in\Re$ is participant $i$'s operational performance:
		\begin{equation}\label{konigstrategy}
			\varphi_i(G,y)=x_i'\beta+y_i+\lambda\sum_{j=1}^n\left(1-\mathbf{1}_\mathcal{D}(j)\right)w^*_{ij}y_j+\epsilon_i,	
		\end{equation}
		where $\lambda\in [0,1)$ is a linear spillover effect from allies fighting efforts, $\mathbf{1}_\mathcal{D}(j)$ is an indicator that takes the value 1 when participant $j$ accepts defeat and takes the cost $D$, and $\epsilon_i$ is an unobserved participant characteristic. As shown in \cite{konig2017networks}, such a model yields a Nash equilibrium which implies the structural relationship (\ref{model}) with $f(\cdot)\equiv 0$, i.e. a linear SAR specification.  
	\end{example}
	\begin{example}[\textit{Network Games with Quadratic Utility}]\label{example:peereffects}
		The model (\ref{model}) $f(\cdot)\equiv 0$ is also implied by a class of canonical games of externalities defined by quadratic utilities, see e.g. \cite{ballester2006s, calvo2009peer}. As an example, suppose a set of $n$ agents is connected via a binary network $W$ where $w_{ij}=1$ if agents $i$ and $j$ are friends, and zero otherwise. Each agent $i$ selects an effort level $e_i$ and receives the payoff 
		\begin{equation}\label{peereffectspayoff}
			u_i(e_i,e_{-i},x_i,W)=e_ix_i'\beta+e_i+\lambda\sum_{j=1}^n w_{ij}e_ie_j-\frac{e_i^2}{2},	
		\end{equation}
		where $e_{-i}$ denotes the effort vector of agents other than $i$, $\lambda\in (-1,1)$ is the peer effect. 
		
		To understand (\ref{peereffectspayoff}) economically, suppose the agents are students and effort is time spent studying for a test. Then (\ref{peereffectspayoff}) is interpreted as agent $i$'s grade is determined by their own observable abilities $x_i$, their own effort $e_i$ and the efforts of their friends $e_{-i}$ which influences their own effort level via the network $W$. 
		
		The Nash equilibrium of this game is given by the effort levels
		\begin{equation}\label{peereffectsNE}
			e^*_i(W)=x_i'\beta+\lambda	\sum_{j=1}^n w_{ij}e_j, i=1,\ldots,n,
		\end{equation}  
		which yields the linear SAR empirical specification. 
	\end{example}
	
	%BROCKDURLAUFJAYARAMAN (Incomplete Information: discuss heuristically using the PINKSESLADE approach)
	
	\begin{example}[\textit{Private Provision of Public Goods in a Network}]
		Consider a version of the model of \cite{bramoulle2007public} with covariates (see also \cite{Bramoulle2014} for further discussion). In such models public goods are local: agent $i$ benefits from or is damaged by their neighbour's good provision. Then, for a network $W$ where $w_{ij}=1$ if agents $i$ and $j$ are neighbours, and zero otherwise, agent $i$'s payoff from taking action $y_i$ is given by
		\begin{equation}\label{publicgoodspayoff}
			u_i\left(y_i,y_{-i},x_i,W\right)=b_i\left(y_i+\lambda\sum_{j=1}^n w_{ij}y_j -  x_i'\beta\right)-\kappa_i y_i,	
		\end{equation}
		where $b_i(\cdot)$ is differentiable, strictly increasing and concave in $y_i$, $\kappa_i$ is marginal cost with $b_i'(0)> \kappa_i> b_i'(\infty)$ and $\lambda\in (0,1)$. This yields the best responses
		\begin{equation}\label{publicgoodsNE}
			y_i^*=x_i'\beta+\bar y_i-\lambda \sum_{j=1}^n w_{ij}y_j,i=1,\ldots,n,
		\end{equation}
		where $\bar y_i$  is defined to be the value that satisfies $b_i^\prime(\bar y_i)\equiv\kappa_i$. The system (\ref{publicgoodsNE}) implies a linear SAR empirical specification.
	\end{example}
	We have presented four examples where the linear SAR specification arose from specific  economic modeling choices, bearing in mind that our test is for linearity. We refer interested readers to Section 2 of \cite{acemoglu2015networks}  for a detailed treatment of  instances where economic modeling leads to nonlinear SAR specifications, namely in networks games, production networks and financial contagion.

	\section{Test statistic}\label{sec:stat}
	We aim to develop a test of  
	\begin{equation}\label{nulltrue}
		\mathcal{H}_0: \  f(x)=0,
	\end{equation}	
	for all $x$ and for some admissible $\lambda$ and $\beta$. Our testing strategy will involve a sieve expansion of $f(\cdot)$ but parameter estimates that can be obtained simply under the null hypothesis (\ref{nulltrue}).
	
	Let $\psi_j(z)$, for $j=1,\ldots,p$, be a user-chosen set of basis functions such that\footnote{Since $x_i$ in (\ref{model}) already contains an intercept, $\psi_j(z)$  does not include a constant term.}
	\begin{equation}\label{lambda_prime}
		f(z)=\sum_{j=1}^p \alpha_j \psi_j(z)+r(z),
	\end{equation}
	with $p=p_n$ being a divergent deterministic sequence, $\alpha= (\alpha_1,\ldots, \alpha_p)^\prime$ a vector of unknown series coefficients and $r(z)$ an approximation error. 	Let $\theta= (\lambda, \beta^\prime)^\prime\in\Theta= \Lambda \times \Re^k$, with $\Lambda$ not necessarily compact but obeying Assumption \ref{ass:Smatrix} below. Define the approximate null hypothesis as
	\begin{equation}\label{approximate_null}
		\mathcal{H}_{0A}: \ \alpha_i=0 \ \ \forall i=1,\ldots,p,
	\end{equation} and the $n\times 1$ vector
	\begin{equation}\label{Spdef}
		S_p(\lambda,\alpha, y)= y-\lambda Wy - \left(\sum_{j=1}^p\alpha_j\psi_j(w_1'y),\ldots,\sum_{j=1}^p\alpha_j\psi_j(w_n'y)\right)'.
	\end{equation}
	Under $\mathcal{H}_{0A}$, we have $S_p(\lambda,0_{p\times 1}, y)= y - \lambda Wy$, consistent with the linear SAR model.

	%	\begin{assumption}\label{ass:apprerror}
	%\textbf{THIS MAY NOT BE NEEDED}		$\sup_{z}r(z)=O_p(p^{-\nu})$, for some $\nu>0$, as $p\rightarrow\infty$.
	%	\end{assumption}

	Our test statistic is based on determining if the moment conditions for the instrumental variables (IV) estimate of $\theta$ under the null hypothesis are close enough to zero. We allow for over-identification, and thus we refer to our estimation method as Two-Stage Least Squares (2SLS) henceforth, see e.g. \cite{kelejian1998generalized}. Given  the expansion in (\ref{lambda_prime}), the approximate IV objective function is
	\begin{equation}\label{2slsobj}
		\mathcal{Q}_p(\lambda, \alpha,\beta, y)= \frac{1}{n}\left( S_p(\lambda, \alpha, y)-X\beta\right)^\prime P_Z\left( S_p(\lambda, \alpha, y)-X\beta\right),
	\end{equation}
	where $Z$ is a $n\times m$ matrix of valid instruments, with $m\geq p+k+1$, $m\sim p$, and $P_Z= Z(Z^\prime Z)^{-1} Z^\prime$.  We will specify the components of the instrument matrix in more detail in the next section. Now, for each $j=1,\ldots,p$, define the $n\times 1$ vector $\Upsilon_j(y)= \left(\psi_j(w_1'y),\ldots,\psi_j(w_n'y)\right)'$
	and the $n\times (p+ k+1) $ matrix  $U= \begin{pmatrix} \Upsilon_{1}(y)& \ldots& \Upsilon_p(y)& Wy & X \end{pmatrix}$ with elements $u_{ij}$. 
	
	To construct our test statistic, we will now introduce the gradient of (\ref{2slsobj}) and its covariance matrix. Define the $(p + k+1)\times 1$ gradient vector $\tilde{d}(\lambda, \beta, y)$ under $\mathcal{H}_{0A}$ as
	\begin{equation}\label{d}
		\tilde{d}(\lambda, \beta, y)= \left.\frac{\partial \mathcal{Q}(\lambda,\alpha, \beta, y)}{\partial (\lambda,\alpha,  \beta)^\prime } \right\vert_{\alpha=0}= -\frac{2}{n}U^\prime P_Z  (y - \lambda W y - X\beta).
	\end{equation}
	We denote by  $\hat{\theta}=(\hat{\lambda}, \hat{\beta}^\prime)^\prime$ the 2SLS estimate of $\theta$ under $\mathcal{H}_{0A}$. Then the gradient evaluated at the residuals corresponding to $\hat\theta$ is
	\begin{equation}\label{dhat}
		\hat{d}= \tilde{d}\left(\hat{\lambda}, \hat{\beta}, y\right)=  -\frac{2}{n}U^\prime P_Z  (y-\hat{\lambda}Wy - X\hat{\beta}).
	\end{equation}
	Likewise, let $\hat{\Sigma}$ be the diagonal matrix with diagonal elements $\hat{\epsilon}_i^2= \left(y_i- \hat{\lambda} \sum_j w_{ij}y_j -x_i^\prime\hat{\beta}\right)^2$, for $i=1,\ldots,n$, 
	$, 
	\hat{J}=n^{-1}Z^\prime U$, where $\hat{J}$ is an $m \times (p+k+1) $, and define the two $m\times m$ matrices $\hat{M}= n^{-1}{Z^\prime Z}$, $\hat{\Omega}=n^{-1}{Z^\prime \hat{\Sigma} Z}.$ The covariance matrix of the gradient evaluated at the estimates is
	\begin{equation}\label{Hhat}
		\hat{H}=\hat{H}\left(\hat{\lambda}, \hat{\beta}, y\right)= 4 \hat{J}^\prime \hat{M}^{-1}\hat{\Omega} \hat{M}^{-1} \hat{J},
	\end{equation}
	and we thence define
	our test statistic as
	\begin{equation}\label{statistic}
		\mathcal{T}= \frac{n\hat{d}^\prime \hat{H}^{-1} \hat{d} - p}{\sqrt{2p}}.
	\end{equation}
	This is a weighted measure of the distance of the gradient from zero, centred and rescaled to account for $p\rightarrow\infty$. 
	
%	Equivalently, we can define the test statistic as 
%	\begin{equation}\label{statistic_equivalent}
%		\mathcal{T}= \frac{n\hat{d}_p^\prime \hat{H}^{11} \hat{d}_p - p}{\sqrt{2p}},
%	\end{equation}
%	where $\hat{d}_p$ contains the first $p$ components of the $(p+k+1)$ vector defined in (\ref{dhat}) and $\hat{H}^{11}$ denotes the top-left $p\times p$ block of $\hat{H}^{-1}$, a notational convention that we maintain for any square matrix considered in the paper.

	\section{Asymptotic theory}\label{sec:asymptotics}
	We commence this section by introducing some technical assumptions to establish the limiting behaviour of (\ref{statistic}) under $\mathcal{H}_{0A}$.
	\begin{assumption}\label{ass:errors}
		For all $n$, $\epsilon _{i}$ 
		are independent random
		variables with zero mean and unknown variance $\sigma_{i} ^{2}\in [k,K] $, $k>0$, and,
		for some $\delta >0,$ $ \mathbb{E}\ |\epsilon _{i}|^{4+\delta }\leq K$ for $i=1,\ldots,n$.
	\end{assumption}
	
	\begin{assumption}\label{ass:regressors}  $ \mathbb{E}(x_{il}^4 )\leq K$, for $i=1,\ldots,n$ and $l=1,\ldots,k$.
	\end{assumption}
	
	The possibility that $cov(\epsilon_i, x_{ij}) \neq 0$, for some $j=1,\ldots, k$, is not ruled out by Assumptions \ref{ass:errors} and \ref{ass:regressors} and thus $X$ might contain some endogenous columns for which external instruments would be needed in $Z$. We denote by $X_1$ the $n \times k_1$ matrix containing the subset of exogenous columns of $X$, while $X_2$ ($n\times k_2$, with $k_2=k-k_1$) contains the  endogenous ones.
	
	%\begin{assumption}\label{ass:truelambda} 
	%	$\lambda _{0}\in \Lambda $, \textit{where }$%
	%	\Lambda $ \textit{is a closed subset in} $(-1,1)$. (EXACTLY WHERE IS THE COMPACTNESS NEEDED?)
	%\end{assumption}
	Now, for a generic symmetric positive-definite matrix $A$, let $\eigbig(A)$ and $\eigsmall(A)$ denote its largest and smallest eigenvalues, respectively. For a generic matrix $B$, denote by $\left\Vert B\right\Vert=\sqrt{\eigbig(B'B)}$, i.e. the spectral norm of $B$, and by  $\left\Vert B\right\Vert_\infty$ its largest absolute row sum.
	
	\begin{assumption}\label{ass:Wmore}
		(i) \textit{For all} $n$\textit{, }$w_{ii}=0$.  (ii) \textit{For all sufficiently large} $n$, $\left\Vert W\right\Vert_{\infty}+
		\left\Vert W^{\prime }\right\Vert_{\infty}\leq K$. 
	\end{assumption}
	
	\begin{assumption}\label{ass:Smatrix} \textit{For all sufficiently large} $n$, $%
		\underset{\lambda \in \Lambda}{\text{sup}} \left(\left\Vert(I -\lambda W)^{-1}\right\Vert_{%
			\infty}+\left\Vert(I-\lambda W^\prime)^{-1}\right\Vert_{\infty}\right) \leq K$.
	\end{assumption}
	
	%\noindent For the next assumption, define the $m\times (p+k+1)$ matrix $\hat{J}=n^{-1} (Z^\prime [\Upsilon_1(y),\ldots. ,\Upsilon_p(y), Wy, X])$ and let $M=\mathbb{E}(\hat{M})$.
	
	\begin{assumption}\label{ass:eigsandinsts} The $m\times m$ matrix $M=\mathbb{E}(\hat{M})$, with $m\geq k+p+1$ , satisfies 
		\begin{equation}\label{ass_6_a}
			\limsup_{n\rightarrow\infty}\eigbig(M) <\infty ,\;\;\; \liminf_{n\rightarrow\infty}\eigsmall\left(M \right) >0
		\end{equation}
		and the $(p+k+1) \times (p+k+1)$ matrix  
		$L=n^{-1}\mathbb{E}(U'U)$ satisfies
		\begin{equation}\label{ass_6_b}
			\limsup_{n\rightarrow\infty}\eigbig(L) <\infty ,\;\;\; \liminf_{n\rightarrow\infty}\eigsmall\left(L \right) >0
		\end{equation}
		for $n$ large enough. For some $\nu>0$ satisfying $n/p^{(\nu + 1/2)}=o(1)$, $\sup_z r(z)=O_p\left(p^{-\nu} \right)$ as $p\rightarrow\infty$.  $\mathbb{E}\left(z_{il_1}^4\right)+\mathbb{E}\left(u_{il_2}^4\right)\leq K$ for $i=1,\ldots,n, l_1=1,\ldots,m$ and $l_2=1,\ldots,p+k+1$, and $\epsilon_i$ and $z_j$ are uncorrelated for each $i,j=1,\ldots,n$.
	\end{assumption}
	
\noindent	Assumptions \ref{ass:Wmore} and \ref{ass:Smatrix} are rather standard restrictions on the spatial weight matrix, helping to control spatial dependence, see e.g. \cite{lee2002consistency, lee2004asymptotic}. Assumption \ref{ass:eigsandinsts} imposes regularity on model primitives and the approximation error and permits conditional heteroskedasticity. In particular, (\ref{ass_6_b}) implies that $\hat{J}$ has full rank $p+k+1$ for sufficiently large $n$, while (\ref{ass_6_a}) is a standard asymptotic boundedness and no-collinearity condition. Our instrument set $Z$ contains, together with at least $k_2$ columns of instruments for the endogenous covariates $X_2$, the columns also of $X_1$, $WX_1$ and a set of instruments $z_{ij}$, $i=1,\ldots,n$. These would be of the form $\psi_q\left(\sum_j w_{ij}x_{1, jl} \right)$, where $q=1,\ldots,p$, and $x_{1, jl}$ denotes the $(j,l)$th element of $X_1$, with $l=1,\ldots,k_1$. Thus, uncorrelatedness between $\epsilon_i$ and $z_j$  requires that $\epsilon_i$ and $x_{1,jl}$ are uncorrelated for all $i,j=1,\ldots, n$ and $l=1,\ldots,k_1$. For more discussion on approximation error decay rates see e.g. \cite{Chen2007}. Finally, we impose
	%Also, (\ref{ass_6_b}) implies that $\hat{\Xi}$ has full rank $p$ for $n$ large enough. We furthermore define the $m\times (k+1)$ matrix $\hat{N}= n^{-1}Z^\prime [Wy, X]$, which has full rank $k+1$ under (\ref{ass_6_b}) in Assumption \ref{ass:eigsandinsts} for $n$ large enough. 
	%NOTE: I MERGED $||J||=O(1)$ WITH 7 (CHECK!). WE NEED TO CROSS REFERENCE CORRECTLY IN THE TEXT.
	\begin{assumption}\label{ass:Mhat}
		Let
		\begin{eqnarray*}
			\Delta_{Z'Z} &=& \underset{1\leq l,k \leq m} {\sup} \underset{j\neq i}{\underset{i=1}{\overset{n}\sum} {{\underset{j=1}{\overset{n}\sum}}}}\left\vert cov(z_{il}z_{ik},z_{jl}z_{jk}) \right\vert,\\
			\Delta_{Z'U} &=& \underset{1 \leq l \leq m,\;\; 1\leq r \leq p+k+1} {\sup} \underset{j\neq i}{\underset{i=1}{\overset{n}\sum} {{\underset{j=1}{\overset{n}\sum}}}}\left\vert cov(z_{il}u_{ir},z_{jl}u_{jr})\right\vert,
		\end{eqnarray*}
		and assume 
		\begin{equation}\label{delta_cond}
			\Delta_{Z'Z}+\Delta_{Z'U} =O(n), \ \ \ \text{as} \ \ \ n\rightarrow \infty.
		\end{equation}
	\end{assumption}
	\noindent Assumption \ref{ass:Mhat} guarantees that a type of weak law of large numbers holds and enforces a suitable bound on cross-sectional dependence in the spirit of \cite{Lee2016}. This can be checked under a variety of conditions such as linear process representations for the underlying random variables or with the near epoch dependence conditions of \cite{Jenish2012}. For example, if the $\{z_{il}\}$ are $L_4$-NED then the $\{z_{il}z_{ik}\}$ would be $L_2$-NED and weak dependence can be guaranteed by a suitable assumption on the NED coefficients. Writing $J=E(\hat J)$, we can now state the following result that is analogous to but stronger than a weak law of large numbers.
	
	\begin{lemma}\label{lemma:Mhat} 
		Let $p^2/n\rightarrow 0$ as $n\rightarrow\infty$ and suppose that Assumptions \ref{ass:regressors}, \ref{ass:eigsandinsts} and \ref{ass:Mhat} hold with $\nu>3/2$. Then, as $n\rightarrow \infty$, 
		\begin{equation}\label{MhatJhat}
			\left\Vert \hat{M} - M \right\Vert =O_p\left(\frac{p}{\sqrt{n}}\right), 		\left\Vert \hat{J} - J \right\Vert =O_p\left(\frac{p}{\sqrt{n}}\right).
		\end{equation}	
	\end{lemma}

	\subsection{Null distribution}\label{subsec:null}

	In this section we develop the asymptotic theory to establish the distribution of the test statistic in (\ref{statistic}) under $\mathcal{H}_{0A}$. Our null asymptotic theory comprises of approximating the test statistic $\mathcal{T}$ with a quadratic form in $\epsilon$, weighted by population quantities. We will then show that this approximation is asymptotically standard normal.  
	To start, we define the population quantities corresponding to (\ref{dhat}), or to its equivalent representation in (\ref{dhat_equiv}), 
	and to (\ref{Hhat}) under $\mathcal{H}_{0A}$ as
	\begin{align}
		d=& -\frac{2}{n} J^\prime M^{-1/2}\left(I - M^{-1/2} N\left(N^\prime M^{-1} N \right)^{-1} N^\prime M^{-1/2} \right) M^{-1/2} Z^\prime \epsilon \notag \\=& -\frac{2}{n} J^\prime M^{-1/2}\mathcal{M}_{NM} M^{-1/2} Z^\prime \epsilon, 
	\end{align}
	where $\mathcal{M}_{NM}= \left(I - M^{-1/2} N\left(N^\prime M^{-1} N \right)^{-1} N^\prime M^{-1/2} \right)$ is $m\times m$ and $N= \mathbb{E}(\hat{N})$, with $\hat{N}= n^{-1} Z^{\prime}(Wy, X)$, which is an $m \times (k+1)$ matrix with full rank under (\ref{ass_6_b}) in Assumption \ref{ass:eigsandinsts},
	and 
	\begin{align}\label{H}
		H= n\mathbb{E}(d d^\prime)  = 4 J^\prime M^{-1/2} \mathcal{M}_{NM} M^{-1/2}\Omega M^{-1/2} \mathcal{M}_{NM} M^{-1/2} J,
	\end{align}
	with $\Omega= n^{-1}\mathbb{E}(Z^\prime \Sigma Z)$ and $\Sigma$ the $n\times n$ diagonal matrix with diagonal given by $\sigma_i^2, i=1,\ldots,n$. Under Assumptions \ref{ass:errors} and \ref{ass:eigsandinsts}, $H^{-1}$ exists and is non-singular for $n$ large enough, via the following lemma for the eigenvalues of $\Omega$. 
	
	\begin{lemma}\label{lemma:Omegaeigs}
		Under Assumptions \ref{ass:errors} and \ref{ass:eigsandinsts}, 
		\[
		\limsup_{n\rightarrow\infty}\eigbig(\Omega)<\infty \text{ and } \liminf_{n\rightarrow\infty}\eigsmall(\Omega)>0.
		\]
	\end{lemma}

	\noindent Also, by construction, the $m\times m$ matrix $\mathcal{M}_{NM}$ has rank $m-k-1$,  which grows like $m$ as $n \rightarrow \infty$ because $k$ is fixed. By relying on the auxiliary Theorem \ref{theorem:dhatd}, reported in Appendix A, we can show that $\mathcal{T}$ can be approximated by a quadratic form in $\epsilon$, as desired. Indeed, we derive
	\begin{theorem}\label{theorem:approx}
		Under Assumptions \ref{ass:errors}-\ref{ass:Mhat},  under $\mathcal{H}_{0A}$ in (\ref{approximate_null}), $\nu>5/2$ and $p^3/n =o(1)$,
		\begin{equation}
			\mathcal{T} - \frac{n d^\prime H^{-1} d - p}{\sqrt{2p}} = o_p(1), \text{ as } n\rightarrow \infty. 
		\end{equation}
	\end{theorem}
	\noindent Thence, we are now ready to prove the main result of this section, which is the asymptotic standard normality of $\mathcal{T}$ under the null hypothesis.
	\begin{theorem}\label{theorem:nulldist}
		Under $\mathcal{H}_0$, Assumptions \ref{ass:errors}-\ref{ass:Mhat}, $\nu>5/2$ and $p^3/n =o(1)$,
		\begin{equation}
			\mathcal{T}\overset{d}{\rightarrow} N(0,1), \text{ as } n\rightarrow \infty.
		\end{equation}
	\end{theorem}
	
	%\noindent %The proofs of Lemma \ref{lemma:Omegaeigs}, Theorems \ref{theorem:approx} and \ref{theorem:nulldist} are in Appendix A.
	\noindent Theorem \ref{theorem:nulldist} provides asymptotic justification for using one-sided, standard normal critical values as observed also by \cite{Hong1995}. However, in practice, a user of the test may be faced with moderate values of $p$. In this situation, the results of the asymptotic test can be compared with a test that employs $(\chi^2_p - p)/\sqrt{2p}$ as the distribution from which its critical values are computed. The comparison will be studied by means of a Monte Carlo experiment in Section \ref{sec:MC}. Observe that one could equivalently just compare the $n\hat d'\hat H ^{-1}\hat d$ to critical values from a $\chi^2_p$ distribution in practice. The standardized $\chi^2_p$ version however allows us to compare the critical values to those from standard normal distribution.

	\subsection{Power properties: Consistency}\label{subsec:powercons}
	
	In this section we establish the consistency of our test, specifically against the alternative
	\begin{equation}
		\mathcal{H}_{1A}: \ \  \alpha_i \neq 0, \ \ \ \text{for some} \ i=1,\ldots,p.
	\end{equation}
	To examine power properties, we introduce the unrestricted quantities
	\begin{equation}\label{epsilonU_def}
		\epsilon_{Ui}(\alpha, \lambda, \beta)=  y_i - \alpha^\prime \begin{pmatrix}\psi_1(w_i^\prime y)\\
			\ldots\\
			\psi_p(w_i^\prime y)\end{pmatrix} - \beta^\prime x_i -\lambda \sum_{j=1}^{n} w_{ij} y_j \ \ \ \text{and} \ \ \ \  \tilde{\Omega}_{U}= \tilde{\Omega}_{U}(\alpha, \lambda, \beta)=n^{-1}{Z^\prime \tilde{\Sigma}_U Z},
	\end{equation}
	where  $\tilde{\Sigma}_{U}$ is a diagonal matrix with $\tilde{\Sigma}_{Uii}= \epsilon_{Ui}^2(\alpha, \lambda, \beta)$. Clearly, $\hat{\Omega}= \tilde{\Omega}_U(0_{p\times 1}, \hat{\lambda}, \hat{\beta})$. Let $\gamma = (\alpha', \lambda, \beta')'\in\Gamma= \Re^p \times \Lambda \times \Re^k$ and introduce:
	\begin{assumption}\label{ass:power} \textit{For all sufficiently large $n$ and all $j=1,\ldots,p+k+1$,}
		\begin{equation}\label{cond_omegaU_1}
			\underset{\gamma \in \Gamma} {\sup} \;\;\eigbig(\tilde{\Omega}_U) =O_p(1) ,\ \ \   \underset{\gamma \in \Gamma} {\sup}\;\; \eigbig\left(\frac{\partial\tilde{\Omega}_U}{\partial \gamma_j}\right) =O_p(1)  ,
		\end{equation}
		and
		\begin{equation}\label{cond_omegaU_2}
			\left\{ \underset{\gamma \in \Gamma} {\inf}\;\;\eigsmall(\tilde{\Omega}_U)\right\}^{-1}=O_p(1),  \ \ \  \left\{ \ \underset{\gamma \in \Gamma} {\inf}\;\;\eigsmall\left(\frac{\partial\tilde{\Omega}_U(\gamma)}{\partial\gamma_j}\right)\right\}^{-1}=O_p(1).
		\end{equation} 
	\end{assumption}

	%The first part of Assumption 1 ensures the following technical lemma (CAN BE MOVED TO APPENDIX) 
	%
	%\begin{lemma}\label{lemma:Omegapowereigs}
	%	Under Assumptions CROSSREF, $ \eigbig(\tilde{\Omega}_U)  =O_p(1)$ (SHOULDN'T THIS BE UNIFORM IN GAMMA TOO?) and $\left\{\underline{\text{eig}}(\tilde{\Omega}_U(\gamma))\right\}^{-1}=O_p(1)$, uniformly in $\gamma \in \Gamma$ (THIS IS A RANDOM EIGENVALUE, RIGHT? I CHANGED TO Op STATEMENT). 
	%\end{lemma} 
	%\begin{proof}[Proof of Lemma \ref{lemma:Omegapowereigs}]
	%	Under the first part of Assumption 10 and Assumptions 6 and 7,
	%\begin{equation}\label{cond_omegaU_1}
	%\underset{\gamma \in \Gamma}{\sup} \ \Vert \tilde{\Omega}_U\Vert \leq  \underset{\gamma \in \Gamma}{\sup} \ \underset{0< i \leq n}{\sup} \ \epsilon_i^2 \left\Vert \frac{Z^\prime Z}{n}\right\Vert = O_p(1). 
	%\end{equation}
	%Now, let $\mathcal{E}=\mathbf{1}\left\{\underset{\gamma \in \Gamma}{\inf}  \underset{0< i \leq n}{\inf} \epsilon_i^2  >0\right\}$, and observe that $\mathcal{E}\rightarrow 1$, by assumption. Then (CAN YOU PLEASE CHECK AS IN THE EARLIER PROOF?) 
	%\begin{align}\label{cond_omegaU_2}
	%\underset{\gamma\in \Gamma}{\inf} \eigsmall (\tilde{\Omega}_U) \geq & \left(\underset{\gamma \in \Gamma}{\inf} \ \underset{0< i \leq n}{\inf} \epsilon_i^2 \eigsmall(\hat{M})\right) \mathcal{E} + \left(\underset{\gamma \in \Gamma}{\inf} \  \underset{0< i \leq n}{\inf} \ \epsilon_i^2 \eigsmall(\hat{M})\right) (1-\mathcal{E}) \notag \\
	%\geq &  k\; \eigsmall(\hat{M}),
	%\end{align}
	%for sufficiently large $n$, whence the result follows.
	%\end{proof}

		\noindent 	 This assumption places some mild regularity on behaviour under the sequence of alternatives $\mathcal{H}_{1A}$, reminiscent of standard boundedness and invertibility conditions. 	We can now state the main theorem of this section.
	
	%We notice that for $\gamma=\gamma_0$, the first part of Assumption 10 reduces to a standard condition of the disturbances $\epsilon_i$ and on the approximation error. We stress that the probability limit of $\Omega_U(\gamma)$ does not need to exist uniformly in $\gamma$. The second part of Assumption 10 is a high-level condition that we need to establish a MVT argument and ensures that bounds such as those in (\ref{cond_omegaU_1}) and (\ref{cond_omegaU_2}) hold for $\partial \tilde{\Omega}_U / \partial \gamma_j$ for all $j=1,\ldots, p+k+1$.

	\begin{theorem}\label{theorem:consistency}
		Under Assumptions \ref{ass:errors}-\ref{ass:power}, $\nu>5/2$ and $p^3/n =o(1)$, $\mathbb{P}(\mathcal{T}>C)\rightarrow 1$ for any $C>0$, as $n\rightarrow\infty$ under $\mathcal{H}_{1A}$.
	\end{theorem}

	\subsection{Power properties: Detection of local alternatives}
	
	We aim to assess the local power of our test by considering the sequence of local alternatives
	\begin{equation}\label{localalt}
		\mathcal{H}_{\ell}: \  \alpha= \alpha_{n}=\frac{p^{1/4}}{\sqrt{n}} \delta, 
	\end{equation}
	where $\delta= (\delta_1,\ldots., \delta_p)^\prime$ with $\Vert \delta \Vert =1$ and $\delta_j$ constants, and the $p^{1/4}$ factor accounts for the cost of the nonparametric approach; see e.g \cite{DeJong1994}, \cite{Hong1995} and \cite{Gupta2023} for a similar dampening factor.  We acknowledge that there are several ways to specify a sequence of local alternatives in the nonparametric testing literature. Our formulation builds on our use of an approximate null hypothesis and relies on local drifts in the series coefficients, as in \cite{Gupta2023}. Other formulations can be found in, \emph{inter alia}, \cite{Kuersteiner2019} and \cite{Anatolyev2023}.
	We also define 
	\begin{equation}
		\mathcal{V}= \begin{pmatrix} I_{p} \\ - (N^\prime M^{-1}N)^{-1} N^\prime M^{-1}\Xi \end{pmatrix} (4J^\prime M^{-1}\Omega M^{-1}J)^{11}\left(I_{p}\ ; \ - \Xi^\prime M^{-1} N (N^\prime M^{-1} N)^{-1} \right),
	\end{equation}
	where $\Xi= \mathbb{E}(\hat{\Xi})$, $\hat{\Xi}= n^{-1} Z^\prime \begin{pmatrix} \Upsilon_{1}(y)& \ldots& \Upsilon_p(y) \end{pmatrix}  $ and $A^{11}$ denotes the top-left $p\times p$ block of ${A}^{-1}$, a notational convention that we maintain for any generic $A^{-1}$ considered in the paper.
	
	\begin{theorem}\label{theorem:local}
		Under $\mathcal{H}_\ell$, Assumptions \ref{ass:errors}-\ref{ass:Mhat}, $\nu>5/2$ and $p^3/n =o(1)$, we have 
		\[
		\mathcal{T}\overset{d}{\rightarrow} N(\varrho,1), \text { as } n\rightarrow\infty,
		\] 
		where $\varrho=\lim_{n\rightarrow\infty}\delta^\prime \Xi^\prime M^{-1} J \mathcal{V} J^\prime M^{-1} \Xi \delta$ is assumed to exist and be positive. 
	\end{theorem}

	\section{Extension to a GMM setup with quadratic moment conditions}\label{sec:quadmoms}
	
	In this section, we show how our methodology can be extended to a GMM setup which employs quadratic moment conditions. One of the advantages of the latter is that it permits identification of the model even when $\beta=0$, and accordingly in this section we permit this. On the other hand, including quadratic moments in the increasing dimension vector of moment conditions complicates matters considerably in establishing the limiting distribution of the test statistic, cf. Lemma A.9 of \cite{Yang2025}. This result, the only one available in the literature for this problem to the best of our knowledge, relies on homoskedasticity and higher order moments. Because of the restrictive nature of these conditions, we provide only the null distribution of the test statistic. In the next section, we apply this in a homoskedastic Monte Carlo design for $\beta=0$, the pure SAR case, and find satisfactory size and power performance.
	
	\begin{assumption}\label{ass:quadmoms}
		The $\epsilon_i$ are i.i.d. with zero mean, variance $\sigma^2\in[k,K]$ and $\mathbb{E}\left(\epsilon_i^{12}\right)<K$.	
	\end{assumption}	
	
		\noindent 	Let $L_{l}, l=1,\ldots,\ell,$ denote a set of $n\times n$ matrices, denote by  $\epsilon_U(\gamma)$  be the vector of the $\epsilon_{Ui}$ and define the $\ell+m$-dimensional vector of GMM moment conditions $q(\gamma) = n^{-1}(L_{1} \epsilon_U(\gamma), \dots, L_{\ell} \epsilon_U(\gamma),Z)' \epsilon_U(\gamma)$ and its $(\ell+m)\times (p+k+1)$ derivative matrix $Q(\gamma)=\frac{\partial q(\gamma)}{\partial \gamma'} = -n^{-1} (L_{1}^{\prime} \epsilon_U(\gamma), \dots, L_{\ell}^{\prime} \epsilon_U(\gamma),Z)' U$, where $\ell,m\rightarrow\infty$, $\ell+m\geq p+k+1$ and $\ell+m\sim p$. For any square $n \times n$ matrix $A$, let $vec_{D}(A) =
	(a_{11},\ldots,a_{nn})^{\prime}$ denote the column vector formed with the diagonal elements of $A$, $A^s=A+A'$, and let $\mu_3 = \mathbb{E}(\epsilon_{i}^3)$ and
	$\mu_4 = \mathbb{E}(\epsilon_{i}^4)$. It follows that, conditional on $Z$ and the $L_i$, $i=1,\ldots,\ell$, $var(\sqrt{n}q(\gamma_{0})) = \hat\Psi$ where
	
	\begin{equation}\label{Psidef}
		\hat\Psi =n^{-1} \left( \begin{array}{cc} (\mu_{4} - 3\sigma^{4})\omega_{\ell}^{\prime}\omega_{\ell} & \mu_{3}\omega_{\ell}^{\prime}Z \\ \mu_{3}Z^{\prime}\omega_{\ell} & 0 \end{array} \right) +n^{-1} V
	\end{equation}
	with $\omega_{\ell} = [vec_D(L_{1}),\ldots,vec_D(L_{\ell})]$ and
	\begin{equation}
		V = \sigma^{4} \left( \begin{array}{cccc} \mathrm{tr}\left(L_{1}L_{1}^{s}\right) & \cdots & \mathrm{tr}\left(L_{1}L_{\ell}^{s}\right) & 0 \\ \vdots & & \vdots & \vdots \\ \mathrm{tr}\left(L_{\ell}L_{1}^{s}\right) & \cdots & \mathrm{tr}\left(L_{\ell}L_{\ell}^{s}\right) & 0 \\ 0 & \cdots & 0 & \frac{1}{\sigma^{4}}Z^{\prime}Z \end{array} \right) = \sigma^{4} \left( \begin{array}{cc} A_{\ell} & 0 \\ 0 & \frac{1}{\sigma^{4}}Z^{\prime}Z \end{array} \right),
	\end{equation}
	where $A_{\ell} = [vec(L_{1}^{\prime}),\ldots,vec(L_{\ell}^{\prime})]^{\prime}[vec(L_{1}^s),\ldots,vec(L_{\ell}^s)]$, by using $\mathrm{tr}(AB) = vec(A^{\prime})^{\prime}$
	$vec(B)$ for any conformable matrices $A$ and $B$, and we use $\gamma_0$ to denote the true value of the parameter $\gamma$.
	
	An optimal GMM estimator employs weighting matrix $\hat\Psi^{-1}$, which is function of the moments $\sigma^2$, $\mu_3$ and $\mu_4$. Denote an initial GMM estimator under the null by $\check{\gamma} = (0_p',\check{\theta}^{\prime} )'$, obtained for example via 2SLS or GMM with identity weighting matrix. Denote $\check{\epsilon}_{i} =  y_{i} - \check{\lambda}\sum_{j=1}^{n}w_{ij}y_{j} -  x_{i}^{\prime}\check{\beta}$. Then $\sigma^2$, $\mu_3$ and $\mu_4$ can be estimated by their sample counterparts constructed using $\check{\epsilon}_{i}$. Plugging these estimates into (\ref{Psidef}), we obtain a consistent (in norm) estimate $\check{\Psi}$ for $\Psi={\mathbb{E}}(\hat{\Psi})$.

	\begin{assumption}\label{ass:Psicheck}
		$\|\check{\Psi} - \Psi\| = O_p\left(p/{\sqrt{n}}+p^{1/2-\nu}\right)$ 
	\end{assumption}

	\begin{assumption}\label{ass:gammacons}
For any $\eta>0$,	$\liminf_{n\rightarrow\infty}\inf_{\theta\in \{\theta^*:\left\Vert\theta^*-\theta_0\right\Vert\geq \eta\}}(\mathbb{E}q(0_p,\theta ))' {\Psi}^{-1} \mathbb{E}(q(0_p,\theta ))>0$, where $\theta_0$ is the true parameter vector that lies in the interior of a compact parameter space $\Theta\subset\Lambda\times \Re^k$.
	\end{assumption}
	\noindent 		Define the $(\ell+m)\times (p+k+1)$ matrix
	\begin{equation}
		D
		= n^{-1}E\begin{pmatrix}
			(\Upsilon(y), Wy)'L_1^s\epsilon & \cdots & (\Upsilon(y), Wy)'L_\ell^s\epsilon & (\Upsilon(y), Wy)' Z \\
			0 & \cdots & 0 & X' Z
		\end{pmatrix}',
	\end{equation}
	where $\Upsilon(y)=(\Upsilon_1(y),\ldots,\Upsilon_p(y))$.
	\begin{assumption}\label{ass:Psieigs} The matrix $\Psi$, with $\ell+m\geq k+p+1$ , satisfies 
		\begin{equation}\label{ass_psi_eig}
			\limsup_{n\rightarrow\infty}\eigbig(\Psi) <\infty ,\;\;\; \liminf_{n\rightarrow\infty}\eigsmall\left(\Psi \right) >0,
		\end{equation}
		and the matrix $D$ satisfies 
		\begin{equation}\label{ass_D_eig}
			\limsup_{n\rightarrow\infty}\eigbig(D'D) <\infty ,\;\;\; \liminf_{n\rightarrow\infty}\eigsmall(D'D ) >0.
		\end{equation}
	\end{assumption}
		\noindent 		Assumption \ref{ass:Psicheck} is equation (45) in the supplement to \cite{Yang2025}, for which sufficient conditions are provided in Remark 3.2 of \cite{Yang2025}. The restricted OGMM estimator $\mathring{\gamma} = (0_p',\mathring{\theta}' )'$ is obtained from the minimization of $q'(0_p,\theta ) \check{\Psi}^{-1} q(0_p,\theta )$. The LM test statistic based on the GMM objective is then defined as
	\begin{equation}
		\mathcal{S} =\frac{n q^{\prime}\!\left(\mathring{\gamma}\right)\check{\Psi}^{-1} Q\!\left(\mathring{\gamma}\right)\!\left[Q^{\prime}\!\left(\mathring{\gamma}\right)\check{\Psi}^{-1} Q\!\left(\mathring{\gamma}\right)\right]^{-1} Q^{\prime}\!\left(\mathring{\gamma}\right)\check{\Psi}^{-1} q\!\left(\mathring{\gamma}\right)-p}{\sqrt{2p}}.
	\end{equation}
	
	\begin{theorem}\label{thm:quad_moms}
		Let Assumptions \ref{ass:regressors}, \ref{ass:Wmore}, \ref{ass:Smatrix}, \ref{ass:eigsandinsts}, \ref{ass:quadmoms}, \ref{ass:Psicheck}, \ref{ass:gammacons} and  \ref{ass:Psieigs} hold with $\nu>5/2$, and $p^3/n\rightarrow 0$ as $n\rightarrow\infty$. Then $\mathcal{S}\overset{d}{\rightarrow}N(0,1)$ under $\mathcal{H}_{0}$.
	\end{theorem}

	\section{Monte Carlo}\label{sec:MC}
\subsection{Linear moments test}
In this section we report the results of a simulation exercise with $k=3$, $\lambda_0=0.4$, $\beta_0=(0.5, -2,1)^\prime$.  The $n \times 3$ matrix $X$ has ones in its first column, while elements in the second and third columns are generated in each replication as i.i.d. random variables   from $U[-2, 2]$ and $U[-2.5,2.5]$, respectively. We generate $\epsilon_i$, for $%
i=1,\ldots,n$, as
\begin{equation}
	\epsilon_{i}=\sigma_{i} \zeta_{i},  \label{epsilon}
\end{equation}
with $\zeta_{i}$ generated either from the standard normal distribution or the t-distribution with 5 degrees of freedom ($t_5$), and employ three mechanisms for the scale parameter $\sigma_i$:  
\begin{itemize}
	\item[a)] Direct construction using the formula 
	\begin{equation}
		\sigma_{i}= \frac{d_{i}}{\sum_{j=1}^{n}{d_{j}}/n},  \label{heteroKP}
	\end{equation}
	where $d_{i}$ denotes the number of
	neighbors of unit $i$, such that, for each generic $W$, $d_i = \text{card}(j:\  w_{ij} \neq 0, i\neq j)$. 
	\item[b)] The $\sigma^2_{i}$ are randomly generated values from a $\chi_2^2$ distribution. 
	\item[c)]  Conditional heteroskedasticity of the form
	$\sigma_{i}= {\sqrt{x_{i2}^2+x_{i3}^2}}/{2}$.
\end{itemize}
For a) the $\sigma_i$ are kept fixed across simulations since $W$ are also kept fixed across iterations as explained below. In  b),  the $\sigma_i$ are kept fixed across  simulations and also across different weight matrix scenarios, while for   c), they  vary across  iterations according to the values of the generated $x_{i2}$ and $x_{i3}$.   The heteroskedasticity design in (%
\ref{heteroKP}) is in line with the simulation work in \cite{kelejian2010specification} and \cite{Arraiz2010}, and is motivated by situations in which
heteroskedasticity arises as units across different regions may have
different numbers of neighbours.  In contrast, the design b) yields error heteroskedasticity that is independent of the spatial dependence implied by the weight matrix while the design c) models conditional heteroskedasticity.

We construct the matrix $Z$ as $n\times (p+k+2)$ matrix with $i$-th row
\[
\left(1,x_{i2}, x_{i3},w_i'x_2,w_i'x_3,\psi_1(w_i'x_{\ell_1}),\psi_2(w_i'x_{\ell_2}),\cdots, \psi_p(w_i'x_{\ell_p})\right)'
\]
with $\ell_1, \ell_4, \ell_5, \ell_8, \ell_9, \ell_{12}=2$ and $\ell_2, \ell_3, \ell_6, \ell_7, \ell_{10},\ell_{11} =3$.   	  We alternate $x_{i2}$ and $x_{i3}$ inside  $\psi_j(\cdot)$ in this way so as to ensure both variables appear in odd and even degree polynomials.  We employ the Hermite polynomials for the basis functions $\psi_j(\cdot)$, starting from the second polynomial.   

We report below  empirical sizes and powers for $n=100, 200, 400, 700, 1000, 2000$ based on $1000$ Monte Carlo replications.  	We set $p$ as the integer part of $n^{1/3}$, leading to $p=4,5,7,8,10, 12$ for  $n=100, 200, 400, 700, 1000, 2000$, respectively, and use a nominal size set at $\alpha=0.05$.
Because our test statistic is a standardized chi-squared type statistic, for small $p$, critical  values based on the standardized $\chi_p^2$ distribution may provide better finite sample approximation than those based on the asymptotic standard normal approximation. Hence we use two critical values, one based on $\chi_p^2$ and defined as $(\chi^2_{p, 1-\alpha}-p)/\sqrt{2p}$ where 
$\chi^2_{p, 1-\alpha}$ is the $1-\alpha$-th percentile of the $\chi_p^2$ distribution, and the other as the $1-\alpha$-th percentile of the standard normal distribution.  For $p=4,5, 7, 8, 10, 12$, the $\chi_p^2$-based critical values are 1.9403,  1.9195, 1.8887, 1.8767, 1.8575 and 1.8424, respectively, significantly larger  than the critical value based on normal distribution, 1.645.  Hence we expect tests based on asymptotic normality to be oversized for small $n$, which is indeed what we observe below.

\subsubsection{Size}
We generate $y$ according to (\ref{model}) under $\mathcal{H}_{0}$ using   
four  different configurations for $W$, motivated by a range of empirical situations.

\begin{itemize}
	\item[1)] Exponential distance:  $w_{ij}=exp(-|\ell _{i}-\ell
	_{j}|)1(|\ell _{i}-\ell _{j}|<\log n)$ where $\ell _{i}$ is location of $i$
	along the interval $[0,n]$ which is generated from i.i.d. $ U[0,n]$.

	\item[2)] Cutoff: $W^*$ is generated as $w_{ij}^*= \Phi(-d_{ij})I(c_{ij}<n^{-2/3})$ if $i   \neq j$, and $w_{ii}^*=0$, where
	$\Phi(  \cdot)$ is the standard normal cdf, $d_{ij}  \sim$ i.i.d. $U[-3,3]$ and $c_{ij} \sim$ i.i.d. $U[0,1]$ and  we set $W=W^*/1.1\|W^*\|$.

	\item[3)] Circulant:  $W$ is a matrix with  $W_{i,i-1}=W_{i,i+1}=0.5, i=1, \ldots, n$.
	
	\item[4)] Random:  $W$ is randomly generated as an $n\times n$ symmetric matrix of zeros and
	ones, where the number of `ones' is set to be the integer part of $2n^{6/5}$.   This matrix  emulates a contiguity matrix.  The average number of non-zero elements in a row ranges 5-9.2 for our choices of $n$.

\end{itemize} 
All $W$ are normalized by their spectral norms, and for  specifications 1), 2) and 4), $W$ are stochastically generated and then fixed across replications.

We present Monte Carlo size results for nominal size of 5$\%$ for  Gaussian and $t_5$ errors in Tables \ref{table:sizeheta} and \ref{table:sizehetb}, respectively. Tests based on the standardized $\chi_p^2$ distribution  lead to good empirical sizes even for small $n$ across all weight matrix designs and heteroskedasticity specifications, while relying on asymptotic normality leads to some oversizing,  as  anticipated.  The extent of oversizing reduces with larger $n$ and $p$, as expected,  but the oversizing still remains at  $n=2000$.  Hence, the standardized $\chi_p^2$-based critical value can provide a useful robustness check to the normal approximation in practice.

\begin{table}
	
	\caption{Monte Carlo size of test of $\mathcal{H}_{0A}$, Gaussian error with heteroskedasticity a), b) and c), using critical values based on $\chi_p^2$ and standard normal distributions}
	\vspace{5pt}
	\centering
	%\begin{small}
	\begin{tabular}{cc cccc  cccc}
		\toprule
		&
		&\multicolumn{4}{c}{{$\chi_p^2$-based}}
		&\multicolumn{4}{c}{{${N}(0,1)$-based}} \\
		
		$\sigma_i$ & $n,p$
		& {expo.}
		& {cutoff}
		& {circ.}
		& {rand.}
		& {expo.}
		& {cutoff}
		& {circ.}
		& {rand.} \\
		
		\cmidrule(rr){1-2}
		\cmidrule(llll){3-6}
		\cmidrule(llll){7-10}
		
		a) & 100, 4 & 0.044 & 0.052 & 0.052 & 0.051 & 0.070 & 0.100 & 0.102 & 0.078 \\
		& 200, 5 & 0.052 & 0.051 & 0.053 & 0.046 & 0.078 & 0.064 & 0.083 & 0.065 \\
		& 400, 7 & 0.051 & 0.050 & 0.048 & 0.050 & 0.083 & 0.063 & 0.068 & 0.065 \\
		& 700, 8 & 0.045 & 0.050 & 0.048 & 0.050 & 0.060 & 0.070 & 0.069 & 0.062 \\
		& 1000, 10 & 0.048 & 0.046 & 0.051 & 0.047 & 0.072 & 0.061 & 0.068 & 0.057 \\
		& 2000, 12 & 0.050 & 0.053 & 0.049 & 0.052 & 0.061 & 0.063 & 0.061 & 0.067 \\
		
		\cmidrule(rr){1-2}
		\cmidrule(llll){3-6}
		\cmidrule(llll){7-10}
		
		b) & 100, 4 & 0.051 & 0.052 & 0.060 & 0.048 & 0.081 & 0.081 & 0.090 & 0.074 \\
		& 200, 5 & 0.050 & 0.045 & 0.053 & 0.049 & 0.072 & 0.071 & 0.078 & 0.073 \\
		& 400, 7 & 0.048 & 0.049 & 0.051 & 0.050 & 0.063 & 0.073 & 0.065 & 0.073 \\
		& 700, 8 & 0.051 & 0.052 & 0.048 & 0.052 & 0.063 & 0.068 & 0.067 & 0.065 \\
		& 1000, 10 & 0.048 & 0.051 & 0.049 & 0.048 & 0.069 & 0.061 & 0.069 & 0.067 \\
		& 2000, 12 & 0.047 & 0.045 & 0.050 & 0.047 & 0.068 & 0.062 & 0.060 & 0.058 \\
		
		\cmidrule(rr){1-2}
		\cmidrule(llll){3-6}
		\cmidrule(llll){7-10}
		
		c) & 100, 4 & 0.066 & 0.058 & 0.075 & 0.062 & 0.091 & 0.088 & 0.110 & 0.087 \\
		& 200, 5 & 0.045 & 0.039 & 0.055 & 0.045 & 0.071 & 0.063 & 0.075 & 0.074 \\
		& 400, 7 & 0.044 & 0.047 & 0.055 & 0.055 & 0.069 & 0.062 & 0.074 & 0.068 \\
		& 700, 8 & 0.044 & 0.046 & 0.042 & 0.045 & 0.063 & 0.055 & 0.068 & 0.057 \\
		& 1000, 10 & 0.048 & 0.042 & 0.057 & 0.045 & 0.063 & 0.058 & 0.074 & 0.062 \\
		& 2000, 12 & 0.047 & 0.056 & 0.050 & 0.051 & 0.064 & 0.067 & 0.069 & 0.070 \\
		
		\bottomrule
	\end{tabular}

	\vspace{10pt}
	\raggedright{\footnotesize Nominal size: $\alpha=0.05$. Columns 3--6 are results from using $\chi_p^2$-based critical values while columns 7--10 are based on using standard normal-based critical values.}
	%\end{small}
	\label{table:sizeheta}
\end{table}

\begin{table}
	
	\caption{Monte Carlo size of test of $\mathcal{H}_{0A}$, $t_5$ error with heteroskedasticity a), b) and c), using critical values based on $\chi_p^2$ and standard normal distributions}
	\vspace{5pt}
	\centering
	
	\begin{tabular}{cc cccc  cccc}
		\toprule
		&
		&\multicolumn{4}{c}{{$\chi_p^2$-based}}
		&\multicolumn{4}{c}{{${N}(0,1)$-based}} \\
		
		$\sigma_i$ & $n,p$
		& {expo.}
		& {cutoff}
		& {circ.}
		& {rand.}
		& {expo.}
		& {cutoff}
		& {circ.}
		& {rand.} \\
		
		\cmidrule(rr){1-2}
		\cmidrule(llll){3-6}
		\cmidrule(llll){7-10}
		
		a) & 100, 4 & 0.064 & 0.069 & 0.064 & 0.068 & 0.097 & 0.106 & 0.097 & 0.100 \\
		& 200, 5 & 0.049 & 0.055 & 0.046 & 0.051 & 0.068 & 0.085 & 0.071 & 0.082 \\
		& 400, 7 & 0.052 & 0.049 & 0.051 & 0.049 & 0.070 & 0.079 & 0.076 & 0.076 \\
		& 700, 8 & 0.050 & 0.049 & 0.047 & 0.048 & 0.065 & 0.052 & 0.069 & 0.065 \\
		& 1000, 10 & 0.046 & 0.048 & 0.050 & 0.045 & 0.059 & 0.068 & 0.072 & 0.060 \\
		& 2000, 12 & 0.047 & 0.047 & 0.051 & 0.050 & 0.063 & 0.056 & 0.064 & 0.070 \\
		
		\cmidrule(rr){1-2}
		\cmidrule(llll){3-6}
		\cmidrule(llll){7-10}
		
		b) & 100, 4 & 0.070 & 0.065 & 0.074 & 0.057 & 0.097 & 0.083 & 0.109 & 0.083 \\
		& 200, 5 & 0.052 & 0.051 & 0.054 & 0.048 & 0.082 & 0.077 & 0.085 & 0.075 \\
		& 400, 7 & 0.049 & 0.052 & 0.048 & 0.051 & 0.075 & 0.074 & 0.070 & 0.071 \\
		& 700, 8 & 0.048 & 0.047 & 0.055 & 0.051 & 0.075 & 0.074 & 0.082 & 0.068 \\
		& 1000, 10 & 0.045 & 0.050 & 0.049 & 0.050 & 0.068 & 0.064 & 0.066 & 0.068 \\
		& 2000, 12 & 0.049 & 0.052 & 0.052 & 0.047 & 0.075 & 0.067 & 0.064 & 0.058 \\
		
		\cmidrule(rr){1-2}
		\cmidrule(llll){3-6}
		\cmidrule(llll){7-10}
		
		c) & 100, 4 & 0.067 & 0.059 & 0.079 & 0.063 & 0.113 & 0.088 & 0.117 & 0.091 \\
		& 200, 5 & 0.063 & 0.054 & 0.067 & 0.063 & 0.093 & 0.075 & 0.089 & 0.086 \\
		& 400, 7 & 0.045 & 0.040 & 0.059 & 0.046 & 0.062 & 0.056 & 0.081 & 0.072 \\
		& 700, 8 & 0.043 & 0.046 & 0.057 & 0.045 & 0.064 & 0.068 & 0.077 & 0.061 \\
		& 1000, 10 & 0.047 & 0.044 & 0.048 & 0.049 & 0.065 & 0.066 & 0.062 & 0.065 \\
		& 2000, 12 & 0.045 & 0.044 & 0.052 & 0.046 & 0.055 & 0.053 & 0.071 & 0.061 \\
		
		\bottomrule
	\end{tabular}

	\vspace{10pt}

	\raggedright{\footnotesize Nominal size: $\alpha=0.05$. Columns 3--6 are results from using $\chi_p^2$-based critical values while columns 7--10 are based on using standard normal-based critical values.}
	
	\label{table:sizehetb}
\end{table}

\subsubsection{Power}

Power analysis requires data generation from a nonlinear model, which in general requires the solving of $n$ nonlinear equations at every iteration, a futile task. We resolve this issue by generating the following two nonlinear SAR models on a lattice, similar to \cite{Jenish2016}.  Letting  $k$ and $j$ denote indices across $\Re^2$ as before and setting $y_{0, j}=0,  j=1,2,\ldots, m_2$ and  $y_{k,0}=0,  k=1,2,\ldots, m_1$:
$$
y_{k,j} = arctan (y_{k-1,j}+y_{k,j-1})+x_{k,j}' \beta +\epsilon_{k,j}, \quad k=1, \ldots, m_1, \quad j=1, \ldots, m_2,
$$
and 
$$
y_{k,j} = log(1+0.25(y_{k-1,j}+y_{k,j-1})^2)+x_{k,j}' \beta +\epsilon_{k,j}, \quad k=1, \ldots, m_1, \quad j=1, \ldots, m_2.
$$

To obtain $n=m_1m_2$ with magnitudes in line with the choices of $n$ of the previous subsection, we use
$(m_1, m_2)=(10,10),  (14,15), (20,20), (26, 27), (31,32), (44, 45)$  leading to 
$n=100, 210, 400, 702, 992, 1980$, respectively. 
We generate the corresponding $W$  with $i=m_2(k-1)+j$, so  that  
$$W(m_2(k-1)+j, m_2(k-2)+j)=W(m_2(k-1)+j, m_2(k-1)+j-1)=1$$ are the nonzero elements of the weight matrix. 
This weight matrix  correctly picks out the relevant neighbours but the misspecification lies in the linearity of SAR model being fitted.  In line with the previous section, we report power results for $p=4, 5, 7, 8, 10, 12$ for $n=100, 210, 400, 702, 992, 1980$ respectively.

Tables  \ref{table:powerheta} and \ref{table:powerhetb} present empirical  power of the test of $\mathcal{H}_{0A}$, employing Gaussian and $t_5$  errors, respectively.  Power is best in general when the error has  conditional heteroskedasticity c), followed by  a) then  b).  This is  expected given the random nature of error heteroskedasticity in b), that  is independent of both the weight matrix and regressors.   The empirical power improves with larger $n$ for both settings of nonlinear SAR, although the improvement is somewhat slow for  $t_5$ error combined with heteroskedasticity of type b).

\begin{table}
	\caption{Monte Carlo power of   test of $\mathcal{H}_{0A}$, Gaussian error with heteroskedasticity a), b) and c), using critical values based on  $\chi_p^2$ and standard normal distributions}
	\centering
	\vspace{5pt}
	\begin{tabular}{cc lclc}
		\toprule
		
		&		&\multicolumn{2}{c}{\footnotesize{$\chi_p^2$-based}  }&\multicolumn{2}{c}{\footnotesize{${N}(0,1)$-based} }\\

		$\sigma_i$		&	$n, p$ & \footnotesize{$log$}  &\footnotesize{$arctan$} & \footnotesize{$log$} &\footnotesize{$arctan$}\\

		\cmidrule(r){1-2}
		\cmidrule(lr){3-4}
		\cmidrule(lr){5-6}
		a)		& 100, 4 & 0.209 & 0.149 & 0.263 & 0.204 \\ 
		& 210, 5 & 0.324 & 0.270 & 0.379 & 0.335 \\ 
		& 400, 7 & 0.525 & 0.426 & 0.594 & 0.494 \\ 
		& 702, 8 & 0.729 & 0.685 & 0.767 & 0.734 \\ 
		& 992, 10 & 0.849 & 0.831 & 0.877 & 0.858 \\ 
		& 1980, 12 & 0.992 & 0.991 & 0.997 & 0.993  \\
		\cmidrule(r){1-2}
		\cmidrule(lr){3-4}
		\cmidrule(lr){5-6}
		b)		&  100, 4 & 0.170 & 0.120 & 0.214 & 0.156 \\ 
		& 210, 5 & 0.202 & 0.148 & 0.243 & 0.202 \\ 
		& 400, 7 & 0.315 & 0.214 & 0.373 & 0.270 \\ 
		& 702, 8 & 0.435 & 0.327 & 0.483 & 0.394 \\ 
		& 992, 10 & 0.598 & 0.457 & 0.647 & 0.518 \\ 
		& 1980, 12 & 0.861 & 0.786 & 0.883 & 0.816 \\
		\cmidrule(r){1-2}
		\cmidrule(lr){3-4}
		\cmidrule(lr){5-6}

		c)      &  100, 4 & 0.307 & 0.223 & 0.349 & 0.291 \\ 
		& 210,5 & 0.439 & 0.36 & 0.488 & 0.41 \\
		& 400, 7 & 0.585 & 0.562 & 0.626 & 0.612 \\ 
		& 702, 8 & 0.827 & 0.817 & 0.855 & 0.845 \\ 
		& 992, 10 & 0.924 & 0.914 & 0.935 & 0.929 \\ 
		& 1980, 12 & 1.000 & 0.997 & 1.000 & 0.997 \\ \bottomrule
		
	\end{tabular}%
	
	\raggedright
	\vspace{10pt}
	{\footnotesize Data  generated as nonlinear SAR.  Columns 3-4 are results from using $\chi_p^2$-based critical values while columns 5-6 are those from using standard normal-based critical values.   }
	\label{table:powerheta}
\end{table}
%$p=4, 5, 7, 8, 10$ for $n=100, 210, 400, 702, 992$, respectively. 

\begin{table}
	\caption{Monte Carlo power of   test of $\mathcal{H}_{0A}$, $t_5$ error with heteroskedasticity a), b) and c), using critical values based on  $\chi_p^2$ and standard normal distributions}
	\centering
	\vspace{5pt}
	\begin{tabular}{cc lclc}
		\toprule
		
		&		&\multicolumn{2}{c}{\footnotesize{$\chi_p^2$-based}  }&\multicolumn{2}{c}{\footnotesize{${N}(0,1)$-based} }\\

		$\sigma_i$		&	$n, p$ & \footnotesize{$log$}  &\footnotesize{$arctan$} & \footnotesize{$log$} &\footnotesize{$arctan$}\\

		\cmidrule(r){1-2}
		\cmidrule(lr){3-4}
		\cmidrule(lr){5-6}
		a)	 & 100, 4  & 0.177 & 0.127 & 0.238 & 0.168 \\
		& 210, 5  & 0.249 & 0.179 & 0.305 & 0.232 \\
		& 400, 7  & 0.348 & 0.297 & 0.395 & 0.35  \\
		& 702, 8  & 0.537 & 0.454 & 0.584 & 0.507 \\
		& 992, 10 & 0.666 & 0.598 & 0.704 & 0.635 \\
		& 1980, 12 & 0.940  & 0.889 & 0.952 & 0.907 \\
		
		\cmidrule(r){1-2}
		\cmidrule(lr){3-4}
		\cmidrule(lr){5-6}
		b)  & 100, 4  & 0.137 & 0.097 & 0.19  & 0.129 \\
		& 210, 5  & 0.162 & 0.12  & 0.21  & 0.156 \\
		& 400, 7  & 0.219 & 0.161 & 0.27  & 0.215 \\
		& 702, 8  & 0.287 & 0.211 & 0.337 & 0.253 \\
		& 992, 10 & 0.379 & 0.258 & 0.423 & 0.302 \\
		& 1980, 12 & 0.642 & 0.470  & 0.690  & 0.517\\
		\cmidrule(r){1-2}
		\cmidrule(lr){3-4}
		\cmidrule(lr){5-6}			
		c)   &   100, 4 & 0.240 & 0.185 & 0.293 & 0.247 \\ 
		& 210,5 & 0.313 & 0.265 & 0.372 & 0.317 \\ 
		& 400, 7 & 0.453 & 0.396 & 0.515 & 0.474 \\ 
		& 702, 8 & 0.679 & 0.605 & 0.728 & 0.654 \\ 
		& 992, 10 & 0.786 & 0.763 & 0.822 & 0.789 \\ 
		& 1980, 12 & 0.981 & 0.956 & 0.985 & 0.968 \\ \bottomrule
	\end{tabular}%
	
	\raggedright
	\vspace{10pt}
	{\footnotesize Data  generated as nonlinear SAR.   Columns 3-4 are results from using $\chi_p^2$-based critical values while columns 5-6 are those from using standard normal-based critical values.   }
	\label{table:powerhetb}
\end{table}

\subsection{Quadratic moments test when $\beta=0$}

In this subsection, we report the empirical size and power of the quadratic moments test using the statistic $\mathcal{S}$ of Theorem \ref{thm:quad_moms}. We are interested in whether a test based on $\mathcal{S}$ performs well in cases not covered by testing with $\mathcal{T}$, and therefore focus on the pure SAR model with $\beta_0 = 0$.

There are two modifications to the simulation setup to satisfy the assumptions of Theorem \ref{thm:quad_moms}. First, we generate the error term $\zeta_i$ in (\ref{epsilon}) from a $t_{14}$ distribution (as well as a standard normal distribution, as before), rather than a $t_5$ distribution, due to the requirement of at least 12th order moments. Second, the error terms are homoskedastic, with $\sigma_i = 1$ in (\ref{epsilon}) for $i = 1, \ldots, n$. We use only quadratic moments involving the matrices $L_1, \ldots, L_{\ell}$, where $\ell$ is set equal to the integer part of $n^{1/3}$ and $L_i=W^i-n^{-1}\text{tr}(W^i)I_n$. 

Tables \ref{table:pureSARsize1} and \ref{table:pureSARsize2} report the size of the test of $\mathcal{H}_{0A}$ under the four weight-matrix specifications considered previously.  The empirical size reported shows  undersizing across the weight matrix specifications that  improves somewhat slowly with increasing $n$ and $p$. This is not suprising given that estimates of $\lambda$ in a pure SAR models often converge slowly, see e.g. \cite{Lee2001} for GMM, which we use, and \cite{lee2004asymptotic} for QMLE. For $n=2000, p=12$, sizing is satisfactory for the test using standard normal critical values for all four weight matrices, but some moderate undersizing remains for the test based on $\chi_p^2$ based critical values.  The patterns are similar between the two tables.

Table \ref{table:pureSARpower} reports empirical power of the test using $\mathcal{S}$  when the data is generated from nonlinear SAR as explained in the previous sub-section.  The empirical power increases with $n$ and $p$, however at a slower rate, especially for the arctan case, so we report results also for larger sample sizes $n=3025, 5041$.\footnote{$m_1, m_2= 55, 71$ respectively.}  The empirical power does increase, albeit slowly, with larger $n,p$ and what is notable is that for the log case, the reported power is identical whether one uses critical values from $\chi_p^2$ or standard normal distributions for larger $n$.

\begin{table}[htbp]
	\centering
	\caption{Monte Carlo size of test of $\mathcal{H}_{0A}$ using $\mathcal{S}$ of Theorem \ref{thm:quad_moms} under pure SAR, Gaussian error, using critical values based on $\chi_p^2$ and standard normal distributions}\label{table:pureSARsize1}
	\label{tab:size_designs}
	\vspace{5pt}
	
	\begin{small}
		\begin{tabular}{c cccc cccc}
			\toprule
			&
			\multicolumn{4}{c}{$\chi_p^2$-based}
			&
			\multicolumn{4}{c}{$N(0,1)$-based}
			\\
			
			\cmidrule(lr){2-5}
			\cmidrule(lr){6-9}
			
			$n,p$
			& expo.
			& cutoff
			& circ.
			& rand.
			& expo.
			& cutoff
			& circ.
			& rand.
			\\
			\midrule
			
			100, 4   & 0.010 & 0.014 & 0.026 & 0.017 & 0.015 & 0.018 & 0.035 & 0.023 \\
			200, 5   & 0.019 & 0.022 & 0.027 & 0.021 & 0.028 & 0.030 & 0.039 & 0.026 \\
			400, 7   & 0.019 & 0.036 & 0.024 & 0.053 & 0.027 & 0.045 & 0.034 & 0.062 \\
			700, 8   & 0.015 & 0.038 & 0.035 & 0.043 & 0.023 & 0.048 & 0.045 & 0.048 \\
			1000, 10 & 0.020 & 0.032 & 0.034 & 0.035 & 0.023 & 0.039 & 0.041 & 0.044 \\
			2000, 12 & 0.038 & 0.052 & 0.036 & 0.034 & 0.044 & 0.056 & 0.052 & 0.040 \\
			
			\bottomrule
		\end{tabular}
		
		\vspace{8pt}
		\raggedright
		\footnotesize
		Nominal size: $\alpha=0.05$. Columns 3--6 are results from using $\chi_p^2$-based critical values while columns 7--10 are based on using standard normal-based critical values.  Based on quadratic moments.
	\end{small}
	
\end{table}

\begin{table}[htbp]
	\centering
	\caption{Monte Carlo size of test of $\mathcal{H}_{0A}$ using $\mathcal{S}$ of Theorem \ref{thm:quad_moms} under pure SAR, $t_{14}$ error, using critical values based on $\chi_p^2$ and standard normal distributions}\label{table:pureSARsize2}
	\label{tab:power_designs}
	\vspace{5pt}
	\begin{small}
		\begin{tabular}{c cccc cccc}
			\toprule
			&
			\multicolumn{4}{c}{$\chi_p^2$-based}
			&
			\multicolumn{4}{c}{$N(0,1)$-based}
			\\
			
			\cmidrule(lr){2-5}
			\cmidrule(lr){6-9}
			
			$n,p$
			& expo.
			& cutoff
			& circ.
			& rand.
			& expo.
			& cutoff
			& circ.
			& rand.
			\\
			\midrule
			
			100, 4   & 0.011 & 0.016 & 0.028 & 0.018 & 0.015 & 0.020 & 0.042 & 0.027 \\
			200, 5   & 0.011 & 0.040 & 0.032 & 0.024 & 0.016 & 0.046 & 0.040 & 0.027 \\
			400, 7   & 0.016 & 0.026 & 0.025 & 0.035 & 0.021 & 0.036 & 0.035 & 0.040 \\
			700, 8   & 0.031 & 0.038 & 0.044 & 0.022 & 0.037 & 0.045 & 0.056 & 0.030 \\
			1000, 10 & 0.027 & 0.041 & 0.033 & 0.038 & 0.031 & 0.048 & 0.041 & 0.044 \\
			2000, 12 & 0.050 & 0.044 & 0.034 & 0.035 & 0.057 & 0.047 & 0.048 & 0.049 \\
			
			\bottomrule
		\end{tabular}
		
		\vspace{8pt}
		\raggedright
		\footnotesize Nominal size: $\alpha=0.05$. Columns 3--6 are results from using $\chi_p^2$-based critical values while columns 7--10 are based on using standard normal-based critical values.  Based on quadratic moments.
	\end{small}
\end{table}

\begin{table}[htbp]
	\centering
	\caption{Monte Carlo power of   test of $\mathcal{H}_{0A}$ using $\mathcal{S}$ of Theorem \ref{thm:quad_moms} under pure SAR, Gaussian and $t_{14}$ error, using critical values based on  $\chi_p^2$ and standard normal distributions}\label{table:pureSARpower}
	\label{tab:power_transform}
	\vspace{5pt}
	\begin{small}
		\begin{tabular}{cc cccc cccc}
			\toprule
			&
			& \multicolumn{4}{c}{Gaussian error}
			& \multicolumn{4}{c}{$t_{14}$ error} \\
			
			\cmidrule(lr){3-6}
			\cmidrule(lr){7-10}
			
			&
			& \multicolumn{2}{c}{log}
			& \multicolumn{2}{c}{arctan}
			& \multicolumn{2}{c}{log}
			& \multicolumn{2}{c}{arctan} \\
			
			\cmidrule(lr){3-4}
			\cmidrule(lr){5-6}
			\cmidrule(lr){7-8}
			\cmidrule(lr){9-10}
			
			$n,p$  
			& 
			& $\chi_p^2$
			& $N(0,1)$
			& $\chi_p^2$
			& $N(0,1)$
			& $\chi_p^2$
			& $N(0,1)$
			& $\chi_p^2$
			& $N(0,1)$ \\
			\midrule
			
			100, 4   & & 0.090 & 0.109 & 0.061 & 0.073 & 0.123 & 0.145 & 0.065 & 0.075 \\
			210, 5   & & 0.264 & 0.303 & 0.097 & 0.117 & 0.277 & 0.313 & 0.089 & 0.101 \\
			400, 7   & & 0.559 & 0.599 & 0.125 & 0.135 & 0.569 & 0.610 & 0.123 & 0.139 \\
			702, 8   & & 0.820 & 0.838 & 0.177 & 0.193 & 0.839 & 0.857 & 0.127 & 0.189 \\
			992, 10 & & 0.887 & 0.890 & 0.195 & 0.218 & 0.903 & 0.909 & 0.184 & 0.206 \\
			1980, 12 & & 0.888 & 0.889 & 0.286 & 0.327 & 0.916 & 0.916 & 0.275 & 0.303 \\
			3025, 14 & & 0.906 & 0.906 & 0.404 & 0.438 & 0.933 & 0.933 & 0.388 & 0.420 \\
			5041, 17 & & 0.915 & 0.915 & 0.635 & 0.659 & 0.948 & 0.949 & 0.573 & 0.604 \\
			
			\bottomrule
		\end{tabular}
		
		\vspace{8pt}
		\raggedright
		\footnotesize Data  generated as nonlinear SAR.    Based on quadratic moments.   Columns labelled $\chi_p^2$ uses critical values based on $\chi_p^2$ distribution while columns labelled $N(0,1)$ use standard normal critical values. 
	\end{small}
\end{table}

	\section{Tax competition between Finnish municipalities}\label{sec:applications}

	Numerous studies have used a linear SAR specification to test for the presence of competitive behaviour in neighbouring governments' tax-setting decisions.  
	While many had found the presence of tax competition based on a spatial IV approach or QML methods (see an extensive list  given in
	\cite{Allers2005}),  \cite{Lyytikaeinen2012}  used a policy-based IV
	to estimate the SAR parameter and found it to be insignificant. 
	In this section we apply our test of linearity to data from  \cite{Lyytikaeinen2012} to try and shed light on this discrepancy.

	We begin with some institutional background. Finland's municipalities  set their own property tax rates
	within limits imposed by the central government. This raises the question of whether neighbouring municipalities compete on tax rates to attract investment. To study this question, \cite{Lyytikaeinen2012} used a linear SAR model with
	fixed effects such that 
	\begin{equation}\label{finnish}
		t_{it}=\lambda \displaystyle\sum_{j=1}^{n}w_{ij}t_{jt}+X_{it}^{\prime }\beta +\mu _{i}+\tau _{t}+\epsilon _{it}  
	\end{equation}%
	where $t_{it}$ denotes either municipality $i$'s general property tax rates or residential building tax rates in year $t$ and $\mu _{i}$ and $\tau _{t}$ are
	municipality and year fixed effects, respectively.  The municipality controls $X_{it}$ include 
	per capita income,  per capita grants,   unemployment rate, 
	percentage of population aged 0-16,  percentage population aged 61-75  and  percentage of population aged 75+.  
	
	\cite{Lyytikaeinen2012} focused on one-year differenced data using the difference between 2000 and 1999 to allow for municipality fixed effects. This choice was to exploit the variation due to a policy of raising  the common statutory lower limit to the property tax rates that was implemented in 2000. Using $\Delta$ to denote a difference between 2000 and 1999, this exogenous policy change was used to construct a suitable instrument and estimate the parameters of   
	\begin{equation}\label{finnish_base}
		\Delta t_{i}=\lambda \displaystyle\sum_{j=1}^{n}w_{ij}\Delta t_{j}+\Delta X_{i}^{\prime }\beta +\gamma_0+\gamma_1 P _{i}+ \gamma_2 M_i + \Delta \epsilon _{i},  \ \ \ \ i=1,\ldots,411,
	\end{equation}%
	where  $P_i$ is a dummy variable indicating whether the 1998 tax rate level for municipality $i$ was below the new lower limit imposed in 2000 and $M_i$ indicates
	the magnitude of the imposed increase for municipality $i$, both included to ensure exogeneity of the instruments.  
	\cite{Lyytikaeinen2012} found the spatial parameter $\lambda $ to be insignificant for both sets of regressions with either general property tax rate or residential building tax rate, and hence concluded that there is an absence of substantial tax competition between municipalities in Finland.  
	
	We apply our test of linearity  in (\ref{finnish_base}) using  $p=4,5,6$.   We estimate the model with the same policy instrument and  row-normalized contiguity weight matrix used in   \cite{Lyytikaeinen2012}.  We   utilize 
	the spatial IV,  constructed by premultiplying $\Delta X_{i}$ by the weight matrix, 
	as our additional instrumental variables required to satisfy Assumption \ref{ass:eigsandinsts},  and report the results in Table \ref{table:lyttest}.   Our tests do not reject the null of linearity for any choice of $p$ for both general property tax rates and residential building tax rates. This indicates that a linear SAR specification is compatible with the differenced data for Finnish municipalities. Using standardized $\chi^2_p$-based critical values evidently does not change our conclusions.

	\begin{table}[H]
		\caption{Linearity test on tax rate data with fixed effects }
		\centering
		\vspace{5pt}
		\begin{tabular}{c cc cc}
			\hline
			&\multicolumn{2}{c}{\footnotesize{General property tax rate}  }&    \multicolumn{2}{c}{\footnotesize{Residential building tax rate}  }  \\ 
			$p$ &$\hat{\lambda}$&  $\mathcal{T}$  &  $\hat{\lambda}$  &  $\mathcal{T}$\\ 
			
			\hline

			4 & $\underset{ (0.9976)}{0.0770}  $  & 0.6005  & $\underset{ (0.4680) }{ 0.0895}$ & -1.3877\\
			5 & $   \underset{(1.0984) }{0.0832} $&  0.3535 & $  \underset{(0.5251)}{0.1024}$ & -0.4511 \\ 
			6 & $   \underset{(1.0886) }{0.0825}$ &    -0.0419 & $ \underset{(0.5398) }{0.1052}$ & -0.2380 \\ 
			\hline
		\end{tabular}%
		
		\raggedright
		\vspace{10pt}
		{ \footnotesize  Using differenced data between 2000 and 1999.  $n=411$.  t-statistics in parenthesis.  * $p-value<0.1$; ** $p-value<0.05$; *** $p-value<0.01$.  }
		\label{table:lyttest}	
	\end{table}
	
	As observed above, the absence of tax competition that \cite{Lyytikaeinen2012} finds differs from earlier findings in the literature. To try and get to the bottom of this, we observe that one notable way in which \cite{Lyytikaeinen2012} differs  from previous literature listed in \cite{Allers2005}  is in the inclusion of municipality fixed effects.  Not accounting for the time-invariant characteristics could result in spurious spatial dependence explaining the disparity of \cite{Lyytikaeinen2012}'s findings from the previous ones. 
	
	With this in mind, we now apply our test of linearity  to the  level data from year 2000, without fixed effects,  given by
	\begin{equation*}
		t_{i}=\lambda \displaystyle\sum_{j=1}^{n}w_{ij} t_{j}+ X_{i}^{\prime }\beta +\gamma_0+\gamma_1 P _{i}+ \gamma_2 M_i + \epsilon _{i},  \ \ \ \ i=1,\ldots,411.
	\end{equation*}
	Interestingly, we observe in Table \ref{table:lyttest2} that the null of linearity is strongly rejected in the case of general property tax rates where the estimated $\lambda$ is positive and significant, in contrast to the case of residential building rate where linearity is not rejected and estimated $\lambda$ is insignificant. Once again, using standardized $\chi^2_p$-based critical values evidently does not change our conclusions. 
	
	Thus the finding of significant spatial competition in the general property tax when not accounting for municipality fixed effects appears to be due to an unreliable specification. This supports the conclusion of \cite{Lyytikaeinen2012} that there is no  competitive behaviour in the setting of Finnish property tax rates and that the linear SAR specification is a reliable model. It also offers further explanation to  \cite{Lyytikaeinen2012}'s insight  that previous findings of the presence of tax competition need to viewed with caution and may be resulting from specification problems.
	
	\section{Conclusion}\label{sec:concl}
	We have presented an easy to implement test for the linearity of the spatial lag in a SAR model. The test is nonparametric but does not require nonparametric estimation due to the reliance on the LM principle. In this sense, it is reminiscent of classical LM diagnostic tests. Our work contributes to a growing literature on specification testing for the SAR model, which has long been subject to scrutiny in this regard, see e.g. \cite{Pinkse2010}. We differ from the existing approaches in assuming an index-type alternative. Other approaches to specification testing include non-smoothing tests, which are under-studied in the SAR context but commonplace in time series and independent data settings, and these are a possible direction for future research.

	\begin{table}[H]
		\caption{Linearity test on tax rate data without fixed effects }
		\vspace{5pt}
		\centering
		\begin{tabular}{c cc cc}
			\hline
			&\multicolumn{2}{c}{\footnotesize{General property tax rate}  }&    \multicolumn{2}{c}{\footnotesize{Residential building tax rate}  }  \\ 
			$p$ &$\hat{\lambda}$&  $\mathcal{T}$  &  $\hat{\lambda}$  &  $\mathcal{T}$\\ 
			
			\hline
			
			4 & $\underset{ (4.0877)  }{0.4019^{***}} $ & $5.8073^{***}$   &  $\underset{ (0.7768) }{0.0955}$ & -0.3789\\
			5 & $  \underset{ (4.2725) }{0.4236^{***}} $&  $4.7876^{***}$ & $  \underset{(0.8404)}{0.1019} $&-0.5596 \\ 
			6 & $ \underset{  (4.2662) }{0.4188^{***}} $&  $4.344^{***}$ & $\underset{ (0.8414)  }{0.1019}$& -1.1595 \\ 
			
			\hline
		\end{tabular}%
		
		\raggedright
		\vspace{10pt}
		{ \footnotesize  Using level data from 2000.   $n=411$.  t-statistics in parenthesis.    * $p-value<0.1$; ** $p-value<0.05$; *** $p-value<0.01$.  }
		\label{table:lyttest2}
	\end{table}

	\newpage
	\section*{Appendix A}
	\renewcommand{\thesection}{A}
	\setcounter{equation}{0} \renewcommand{\theequation}{A.%
		\arabic{equation}} 
	\setcounter{theorem}{0} \renewcommand{\thetheorem}{A\arabic{theorem}}
	
	\begin{proof}[Proof of Lemma \ref{lemma:Mhat}]
		$\left\Vert \hat{M} - M\right\Vert$ has variance bounded by 
		\begin{align}
			\frac{m^2}{n^2} \Delta_{Z'Z} + \frac{m^2}{n^2} \underset{l,k}{\sup} \underset{i}{\sum} \mathbb{E}\left(z_{il}^2 z_{ik}^2\right) &\leq \frac{m^2}{n^2} \Delta_{Z'Z} + \frac{m^2}{n^2} \underset{l,k}{\sup}\underset{i}{\sum} \left( \mathbb{E}(z_{il}^4)\right)^{1/2}\left( \mathbb{E}(z_{ik}^4)\right)^{1/2}\notag\\
			& = O\left(\frac{p^2}{n}\right) + O\left(\frac{p^2}{n}\right),
		\end{align}
		where the first equality holds under (\ref{delta_cond}), since $m\sim p$,  and the second one follows as long as $p/\sqrt{n}=o(1)$ under Assumption \ref{ass:eigsandinsts}. Next, $\left\Vert \hat{J} - J\right\Vert$ has variance bounded by 
		\begin{align}
			\frac{m^2}{n^2} \Delta_{Z'U} + \frac{m^2}{n^2} \underset{l,r}{\sup} \underset{i}{\sum} \mathbb{E}\left(z_{il}^2 u_{ir}^2\right) &\leq \frac{m^2}{n^2} \Delta_{Z'U} + \frac{m^2}{n^2} \underset{l,r}{\sup}\underset{i}{\sum} \left( \mathbb{E}(z_{il}^4)\right)^{1/2}\left( \mathbb{E}(u_{ir}^4)\right)^{1/2}\notag\\
			& = O\left(\frac{p^2}{n}\right) + O\left(\frac{p^2}{n}\right),
		\end{align}
		where the first equality holds under (\ref{delta_cond}), since $m\sim p$,  and the second one follows as long as $p/\sqrt{n}=o(1)$ under Assumption \ref{ass:eigsandinsts}.
	\end{proof}
	
	\begin{theorem}\label{theorem:dhatd} 
		Under Assumptions \ref{ass:errors}-\ref{ass:Mhat}, under $\mathcal{H}_{0A}$ in (\ref{approximate_null}), for $p^{3/2}/n \rightarrow 0$ as $n\rightarrow \infty$,
		\begin{equation}
			\left\Vert \hat{d} - d \right \Vert =O_p\left( \frac{p^{3/2}}{n} \right).
		\end{equation}
		
	\end{theorem}
	
	\begin{proof}
		
		\noindent  We first show the following rate for the IV/2SLS estimator under $\mathcal{H}_{0A}$ in (\ref{approximate_null})
		\begin{equation}\label{IV_rate}
			\left\Vert \hat{\theta} - \theta \right\Vert = O_p\left(\sqrt{\frac{p}{n}}\right),
		\end{equation}
		where $\mathbb{X}= (Wy, X)$. We write, under Assumptions \ref{ass:regressors}, \ref{ass:eigsandinsts} and \ref{ass:Mhat},
		\begin{align}\label{IV_rate1}
			\left \Vert \hat{\theta}- \theta\right \Vert =& \left\Vert\left(\frac{1}{n}\mathbb{X}^\prime P_Z \mathbb{X}\right)^{-1} \frac{1}{n} \mathbb{X}^\prime P_Z \epsilon \right\Vert \leq K \left\Vert \frac{Z^\prime \epsilon}{n}\right \Vert = O_p\left(\sqrt{\frac{p}{n}} \right),
		\end{align}
		since under Assumptions \ref{ass:errors} and \ref{ass:eigsandinsts}, $E\left\Vert Z^\prime \epsilon/n\right\Vert^2$ is bounded by
		\begin{equation}\label{theorem1_1}
			\frac{K}{n^2} m \underset{0 < j \leq p}{\sup} \ \mathbb{E}(z_j^\prime \epsilon)^2 =\frac{K}{n^2} m \underset{0 < j \leq m}{\sup} \ \underset{t}\sum\mathbb{E}(z_{tj}^2)\mathbb{E}(\epsilon_t^2) = O\left(\frac{m}{n}\right)=O\left(\frac{p}{n}\right).
		\end{equation}
		
		Let $R=\left( r(w_1' y), \ldots ,r(w_n' y)  \right)^\prime$ be the $n\times 1$ vector of approximation errors and (\ref{lambda_prime}). We denote its $i-$th component by $R_i$. Also, from the 2SLS expression for $\hat{\theta}- \theta $ in (\ref{IV_rate1}),
		\begin{align}\label{dhat_equiv}
			\hat{d}=& -\frac{2}{n}U^\prime  P_Z \left(I - \mathbb{X} (\mathbb{X}^\prime P_Z \mathbb{X})^{-1} \mathbb{X}^\prime P_Z \right) \epsilon  -\frac{2}{n}U^\prime  P_Z R \notag \\
			=& -\frac{2}{n}U^\prime P_Z \left(I - P_Z \mathbb{X} (\mathbb{X}^\prime P_Z \mathbb{X})^{-1} \mathbb{X}^\prime P_Z \right) P_Z \epsilon  -\frac{2}{n} U^\prime  P_Z R \notag \\
			=& -\frac{2}{n} \hat{J}^\prime \hat{M}^{-1/2}\left(I - \hat{M}^{-1/2} \hat{N}\left( \hat{N}^\prime \hat{M}^{-1} \hat{N} \right)^{-1} \hat{N}^\prime \hat{M}^{-1/2} \right) \hat{M}^{-1/2} Z^\prime \epsilon - \frac{2}{n} \hat{J}^\prime \hat{M}^{-1} Z^\prime R \notag \\
			=& -\frac{2}{n} \hat{J}^\prime \hat{M}^{-1/2} \hat{\mathcal{M}}_{NM} \hat{M}^{-1/2} Z^\prime \epsilon - \frac{2}{n} \hat{J}^\prime \hat{M}^{-1} Z^\prime R, 
		\end{align}
		where $\hat{\mathcal{M}}_{NM}= \left(I - \hat{M}^{-1/2} \hat{N}\left( \hat{N}^\prime \hat{M}^{-1} \hat{N} \right)^{-1} \hat{N}^\prime \hat{M}^{-1/2} \right)$. From (\ref{dhat}), we write
		\begin{align}\label{d_decomposition}
			\left\Vert \hat{d} - d\right \Vert  \leq \left\Vert \frac{2}{n} \hat{J}^\prime \hat{M}^{-1/2}\hat{\mathcal{M}}_{NM}\hat{M}^{-1/2} Z^\prime \epsilon  - \frac{2}{n} J^\prime M^{-1/2}\mathcal{M}_{NM} M^{-1/2} Z^\prime \epsilon  \right \Vert +\left\Vert \frac{2}{n} \hat{J}^\prime \hat{M}^{-1} Z^\prime R\right\Vert. 
		\end{align}
		
	\noindent	By standard algebra, the first term in (\ref{d_decomposition}) is bounded by
		\begin{align}
			&\left\Vert \hat{J} - J \right\Vert \left\Vert \hat{M}^{-1} \right\Vert \left\Vert \frac{1}{n} Z^\prime \epsilon \right\Vert + \left\Vert J\right\Vert \left\Vert \hat{M}^{-1} - M^{-1}\right \Vert \left\Vert \frac{1}{n}Z^\prime \epsilon \right\Vert  \notag\\
			&+\left\Vert \hat{J}  - J\right \Vert \left\Vert \hat{M}^{-1}\right \Vert \left\Vert \hat{N}\right \Vert \left\Vert \left( \hat{N}^\prime \hat{M}^{-1} \hat{N}\right)^{-1}\right \Vert \left\Vert \hat{N} \right \Vert  \left\Vert \hat{M}^{-1} \right \Vert \left\Vert \frac{1}{n}Z^\prime \epsilon\right \Vert \notag \\
			+& \left\Vert J \right \Vert \left\Vert \hat{M}^{-1} -M^{-1}\right \Vert \left\Vert \hat{N}\right \Vert \left\Vert \left( \hat{N}^\prime \hat{M}^{-1} \hat{N}\right)^{-1}\right \Vert \left\Vert \hat{N} \right \Vert  \left\Vert \hat{M}^{-1}\right \Vert \left\Vert \frac{1}{n}Z^\prime \epsilon\right \Vert \notag \\ +& \left\Vert J \right \Vert \left\Vert M^{-1}\right \Vert \left\Vert \hat{N} - N\right \Vert \left\Vert \left( \hat{N}^\prime \hat{M}^{-1} \hat{N}\right)^{-1}\right \Vert \left\Vert \hat{N} \right \Vert  \left\Vert \hat{M}^{-1}\right \Vert \left\Vert \frac{1}{n}Z^\prime \epsilon\right \Vert  \notag \\
			+& \left\Vert J \right \Vert \left\Vert M^{-1}\right \Vert \left\Vert N\right \Vert \left\Vert \left( \hat{N}^\prime \hat{M}^{-1} \hat{N}\right)^{-1} - \left(N^\prime M^{-1} N\right)^{-1} \right \Vert \left\Vert \hat{N} \right \Vert  \left\Vert \hat{M}^{-1}\right \Vert \left\Vert \frac{1}{n}Z^\prime \epsilon\right \Vert \notag \\
			+& \left\Vert J \right \Vert \left\Vert M^{-1}\right \Vert \left\Vert N\right \Vert \left\Vert  \left(N^\prime M^{-1} N\right)^{-1} \right \Vert \left\Vert \hat{N} -N  \right \Vert  \left\Vert \hat{M}^{-1}\right \Vert \left\Vert \frac{1}{n}Z^\prime \epsilon\right \Vert \notag \\+& \left\Vert J \right \Vert \left\Vert M^{-1}\right \Vert \left\Vert N\right \Vert \left\Vert  \left(N^\prime M^{-1} N\right)^{-1} \right \Vert \left\Vert N  \right \Vert  \left\Vert \hat{M}^{-1} - M^{-1}\right \Vert \left\Vert \frac{1}{n}Z^\prime \epsilon\right \Vert 
		\end{align}
		From (\ref{theorem1_1}), $\left\Vert Z^\prime \epsilon /n\right\Vert = O_p(\sqrt{p/n})$. 
		Under Assumption \ref{ass:Mhat}, we have
		\begin{equation}
			\left\Vert \hat{N} - N\right \Vert =O_p\left(\frac{p}{\sqrt{n}}\right) \ \ \ \ \text{and} \ \ \ \ \  \left\Vert \hat{J} - J \right \Vert =O_p\left(\frac{p}{\sqrt{n}}\right) .
		\end{equation}
		Also, under Assumptions \ref{ass:eigsandinsts} and \ref{ass:Mhat}, 
		\begin{equation}
			\left\Vert \hat{M}^{-1}- M^{-1} \right \Vert \leq \left\Vert M^{-1}\right \Vert  \left\Vert \hat{M}^{-1}\right \Vert  \left\Vert \hat{M}- M \right \Vert = O_p\left(\frac{p}{\sqrt{n}}\right)
		\end{equation}
		and similarly, using analogous arguments (omitted to avoid repetitions), under Assumptions  \ref{ass:eigsandinsts} and \ref{ass:Mhat}, 
		\begin{align}
			\left\Vert (\hat{N}^\prime \hat{M}^{-1} \hat{N})^{-1}  - (N^\prime M^{-1} N)^{-1}\right\Vert &\leq \left\Vert N^\prime M^{-1} N \right \Vert  \left\Vert \hat{N}^\prime \hat{M}^{-1} \hat{N} \right \Vert \left\Vert \hat{N}^{\prime} \hat{M}^{-1} \hat{N} - N^\prime M^{-1} N\right \Vert \notag\\
			&= O_p\left(\frac{p}{\sqrt{n}}\right).
		\end{align}
				\noindent The first term at the RHS of (\ref{d_decomposition}) is thus $O_p\left( p^{3/2}/n \right)$. The second term at the RHS of (\ref{d_decomposition}) is instead
		\begin{equation}\label{d_remainder}
			\left\Vert \frac{2}{n}\hat{J}^\prime \hat{M}^{-1} Z^\prime R \right \Vert = O_p\left(\frac{1}{n} \left\Vert R\right \Vert \left\Vert Z\right \Vert \right) = O_p(p^{-\nu}),
		\end{equation}
		where the first equality at the RHS of (\ref{d_remainder}) follows under Assumptions \ref{ass:eigsandinsts} and \ref{ass:Mhat}. The second equality follows since $\Vert Z \Vert = O_p(\sqrt{n})$ under Assumptions  \ref{ass:eigsandinsts} and \ref{ass:Mhat}, and each component of the $n\times 1$ vector $R$ is $	 O_p(p^{-\nu})$
		by Assumption \ref{ass:eigsandinsts}, and hence $||R||= O_p(\sqrt{n} \  p^{-\nu})$. The third equality in (\ref{d_remainder}) follows from Assumption \ref{ass:eigsandinsts}.
		
		Under Assumption \ref{ass:Mhat}, the first term in (\ref{d_decomposition}) dominates the second one as long as $\nu$ satisfies $n/p^{\nu + 3/2}=o(1)$ as $n\rightarrow \infty$, which holds under Assumptions \ref{ass:eigsandinsts}.
	\end{proof}

	\begin{proof}[Proof of Lemma \ref{lemma:Omegaeigs}]
		First, under Assumptions \ref{ass:errors} and \ref{ass:eigsandinsts},
		\begin{align}\label{omega1}
			\limsup_{n\rightarrow\infty}\eigbig(\Omega)&= \limsup_{n\rightarrow\infty}\Vert \Omega\Vert = \limsup_{n\rightarrow\infty}\Vert n^{-1} \mathbb{E}(Z^\prime \Sigma Z)\Vert \leq  \underset{i}{\sup}\  \sigma^2_i \ \limsup_{n\rightarrow\infty} \left\Vert \frac{\mathbb{E}(Z^\prime Z)}{n}\right \Vert \notag \\ &=\underset{i}{\sup}\  \sigma^2_i \Vert M \Vert<K.
		\end{align}
		Next,
		\begin{align}\label{omega2}
			\liminf_{n\rightarrow\infty}\eigsmall(\Omega)&=\liminf_{n\rightarrow\infty}\eigsmall\left(\mathbb{E}\left(\frac{Z^\prime \Sigma Z}{n}\right)\right)=\liminf_{n\rightarrow\infty}\eigsmall\left(\mathbb{E}\left(\frac{\sum_i z_i z_i^\prime \sigma_i^2}{n}\right)\right) \notag \\& \geq  \underset{i}{\inf} \ \sigma^2_i\ \liminf_{n\rightarrow\infty}  \eigsmall(M)>k,
		\end{align}
		again under Assumptions \ref{ass:errors} and \ref{ass:eigsandinsts}.
	\end{proof}
	
	\begin{proof}[Proof of Theorem \ref{theorem:approx}]
		
		\noindent The claim in Theorem \ref{theorem:approx} is equivalent to 
		\begin{equation}\label{th2_equiv}
			\hat{d}^\prime \hat{H}^{-1}\hat{d} -d^\prime H^{-1}d= o_p\left(\frac{\sqrt{p}}{n}\right).
		\end{equation}
		From some standard manipulations, the LHS of (\ref{th2_equiv}) can be written as
		\begin{equation}\label{th2_terms}
			\left(\hat{d} -d \right)^\prime \hat{H}^{-1}\hat{d} + d^\prime H^{-1} (\hat{d}-d) + d^\prime \hat{H}^{-1}\left(H- \hat{H}\right) H^{-1}\hat{d}.
		\end{equation}
		The norm of the last displayed expression is bounded by
		\begin{equation}\label{theorem2_1}
			K\left\Vert \hat{d} - d \right \Vert \left\Vert \hat{H}^{-1} \right \Vert \left\Vert \hat{d}\right \Vert + K\left\Vert \hat{d} - d \right \Vert \left\Vert H^{-1} \right \Vert \left\Vert d\right \Vert + K \left\Vert d\right \Vert  \left \Vert  \hat{H}^{-1} \right \Vert \left \Vert H- \hat{H}\right \Vert \left \Vert H^{-1}\right \Vert \left\Vert \hat{d}\right \Vert .
		\end{equation}
		From Theorem \ref{theorem:dhatd}, $\left\Vert \hat{d} -d \right \Vert = O_p(p^{3/2}/n)$. 
		Under Assumptions \ref{ass:eigsandinsts} and \ref{ass:Mhat}, and from (\ref{theorem1_1}) we have $\left\Vert d\right \Vert = O\left(\sqrt{p/n}\right)$. Also, from Theorem 1,
		\begin{equation}
			\left\Vert \hat{d}\right \Vert \leq \left\Vert \hat{d}-d\right \Vert + \left\Vert d\right \Vert = O_p\left(\sqrt{\frac{p}{n}}\right),
		\end{equation}
		where the last equality follows as long as $p^2/n =o(1)$.
		Also, under Assumptions \ref{ass:errors},  \ref{ass:eigsandinsts} and \ref{ass:Mhat}, $\left\Vert \hat{H}^{-1} \right \Vert=O_p(1)$ and $\left\Vert H^{-1} \right \Vert=O_p(1)$. Thus,  first and second terms in (\ref{theorem2_1}) are $O_p(p^2/n^{3/2})$, which are $o_p(p^{1/2}/n)$ if $p^3/n = o(1)$.
		
		Using 2SLS estimates for $\theta_0$ and in view of (\ref{dhat_equiv}), we can re-write the expression in (\ref{Hhat}) as
		\begin{align}\label{Hhat_equiv}
			\hat{H}&= 4 \hat{J}^\prime \hat{M}^{-1/2} \hat{\mathcal{M}}_{NM} \hat{M}^{-1/2} \tilde{\Omega}\hat{M}^{-1/2} \hat{\mathcal{M}}_{NM} \hat{M}^{-1/2} \hat{J} + \frac{4}{n} \hat{J}^\prime \hat{M}^{-1/2}\hat{\mathcal{M}}_{NM} \hat{M}^{-1/2} Z^\prime \epsilon R^\prime Z \hat{M}^{-1} \hat{J} \notag \\
			&+ \frac{4}{n}\left( \hat{J}^\prime \hat{M}^{-1/2}\hat{\mathcal{M}}_{NM} \hat{M}^{-1/2} Z^\prime \epsilon R^\prime Z \hat{M}^{-1} \hat{J} \right)'\frac{4}{n} \hat{J}^\prime \hat{M}^{-1} Z^\prime R R^\prime Z \hat{M}^{-1} \hat{J}
		\end{align}
		where $\tilde{\Omega}=Z^\prime \tilde{\Sigma} Z/n$, with $\tilde{\Sigma}$ being an $n\times n$  diagonal matrix such that $\tilde{\Sigma}_{ii}=\epsilon_i ^2$. From (\ref{H}), we write
		\begin{align}\label{H_conv}
			\left \Vert \hat{H} - H \right \Vert  \leq & \left\Vert \hat{J}^\prime \hat{M}^{-1/2} \hat{\mathcal{M}}_{NM} \hat{M}^{-1/2} \tilde{\Omega}\hat{M}^{-1/2} \hat{\mathcal{M}}_{NM} \hat{M}^{-1/2} \hat{J} \right.\notag\\
			& \left.- J^\prime M^{-1/2} \mathcal{M}_{NM} M^{-1/2} \Omega M^{-1/2}\mathcal{M}_{NM} M^{-1/2} J \right\Vert \notag \\
			+& 2\left\Vert \frac{4}{n} \hat{J}^\prime \hat{M}^{-1/2}\hat{\mathcal{M}}_{NM} \hat{M}^{-1/2} Z^\prime \epsilon R^\prime Z \hat{M}^{-1} \hat{J} \right\Vert + \left\Vert \frac{4}{n} \hat{J}^\prime \hat{M}^{-1} Z^\prime R R^\prime Z \hat{M}^{-1} \hat{J}
			\right \Vert.
		\end{align}
		By standard algebra, the first term in (\ref{H_conv}) is bounded by
		\begin{align}
			\begin{split}
				&\left\Vert \hat{J}^\prime \hat{M}^{-1} \tilde{\Omega} \hat{M}^{-1} \hat{J} -  J^\prime M^{-1} \Omega M^{-1} J \right\Vert \notag \\+ &\left\Vert \hat{J}^\prime \hat{M}^{-1} \hat{N} \left(\hat{N}^\prime \hat{M}^{-1} \hat{N}\right)^{-1} \hat{N}^\prime \hat{M}^{-1} \tilde{\Omega}\hat{M}^{-1}\hat{J} - J^\prime M^{-1} N \left(N^\prime M^{-1} N\right)^{-1} N^\prime M^{-1} \Omega M^{-1} J \right\Vert \notag \\
				+& \left\Vert \hat{J}^\prime \hat{M}^{-1} \hat{N} \left(\hat{N}^\prime \hat{M}^{-1} \hat{N}\right)^{-1} \hat{N}^\prime \hat{M}^{-1} \tilde{\Omega}\hat{M}^{-1} \hat{N} \left(\hat{N}^\prime \hat{M}^{-1} \hat{N} \right)^{-1} \hat{N} \hat{M}^{-1}\hat{J} \right. \notag \\- & \left.  J^\prime M^{-1} N \left(N^\prime M^{-1} N\right)^{-1} N^\prime M^{-1} \Omega M^{-1} N \left(N^\prime M^{-1} N \right)^{-1} N M^{-1}J \right \Vert .
			\end{split}
		\end{align}
		\normalsize
		We provide the details of the rate of the first term in the last displayed expression, while the remaining ones follow similarly. Specifically,
		\begin{align}\label{H_conv_first_1}
			&\left\Vert \hat{J}^\prime \hat{M}^{-1} \tilde{\Omega} \hat{M}^{-1} \hat{J} -  J^\prime M^{-1} \Omega M^{-1} J \right\Vert  \leq  \left \Vert \hat{J}-J \right\Vert \left\Vert \hat{M}^{-1}\right\Vert ^2 \left\Vert \tilde{\Omega}\right\Vert  \left\Vert \hat{J}\right\Vert \notag \\+& \left \Vert J \right\Vert \left\Vert \hat{M}^{-1}- M^{-1}\right\Vert \left\Vert \tilde{\Omega}\right\Vert \left\Vert \hat{M}^{-1}\right\Vert \left\Vert \hat{J}\right\Vert + \left \Vert J \right\Vert \left\Vert M^{-1}\right\Vert \left\Vert \tilde{\Omega}- \Omega \right\Vert \left\Vert \hat{M}^{-1}\right\Vert \left\Vert \hat{J}\right\Vert \notag \\
			+& \left \Vert J \right\Vert \left\Vert M^{-1}\right\Vert \left\Vert \Omega\right\Vert \left\Vert \hat{M}^{-1}- M^{-1}\right\Vert \left\Vert \hat{J}\right\Vert + \left \Vert J \right\Vert \left\Vert M^{-1}\right\Vert^2 \left\Vert \Omega\right\Vert \left\Vert \hat{J} - J\right\Vert.
		\end{align}
		Under Assumptions  \ref{ass:eigsandinsts} and \ref{ass:Mhat}, most terms can be dealt with similarly to the proof of Theorem \ref{theorem:dhatd}, and $||\hat{M}^{-1}-M^{-1}||=O_p(p/\sqrt{n})$ and $||\hat{J}- J ||=O_p(p/\sqrt{n})$. We focus instead on $||\tilde{\Omega}- \Omega||$ and write
		\begin{equation}\label{omega_terms}
			\left\Vert \tilde{\Omega} -\Omega \right \Vert \leq \left\Vert \tilde{\Omega} -\bar{\Omega}\right\Vert + \left\Vert \bar{\Omega} - \Omega \right\Vert,
		\end{equation}
		with $\bar{\Omega}= Z^\prime \Sigma Z /n$. The first term in (\ref{omega_terms}) is $\left\Vert Z^\prime (\tilde{\Sigma}- \Sigma) Z/n \right\Vert$, which has second moment bounded by
		\begin{align}
			\frac{m^2}{n^2}\underset{1\leq i,j \leq n}{\sup} \underset{l}\sum \underset{k}\sum \mathbb{E}(z_{li} z_{lj} z_{ki}z_{kj}) \mathbb{E} \left((\epsilon_l^2 - \sigma_l^2)( \epsilon_k^2 - \sigma_k^2)\right)  &\leq  \frac{Km^2}{n^2}\underset{1\leq i,j \leq n}{\sup} \underset{l}\sum \mathbb{E}(z_{li}^2 z_{lj}^2)\notag\\
			& = O\left(\frac{p^2}{n}\right),
		\end{align}
		under Assumptions \ref{ass:errors} and \ref{ass:eigsandinsts} and since $m\sim p$. Thus, $ \left\Vert \tilde{\Omega} - \bar{\Omega}\right\Vert =O_p(p/\sqrt{n})$.
		The second term in (\ref{omega_terms}) has mean zero and variance bounded by 
		\begin{equation}\label{omega_terms2}
			\frac{m^2}{n^2} \underset{1\leq i,j \leq n}{\sup} \ var\left(\underset{l} \sum z_{li}z_{lj} \sigma_l^2\right) \leq \underset{1\leq l \leq n} {\sup} \ \sigma_l^2 \frac{m^2}{n^2}\Delta = O\left(\frac{p^2}{n}\right), 
		\end{equation}
		under Assumptions  \ref{ass:errors}, \ref{ass:eigsandinsts} and \ref{ass:Mhat}, and hence it is $O_p(p/\sqrt{n})$. Thus, $||\tilde{\Omega}- \Omega||= O_p(p/\sqrt{n})$. 
		We conclude then that the expression in (\ref{H_conv_first_1}) is $O_p\left(p/\sqrt{n}\right)$. Proceeding similarly for all other terms, we can conclude that the first term at the RHS of  (\ref{H_conv}) is $O_p\left(p/\sqrt{n}\right)$. 
		
		By similar arguments to those that led to (\ref{d_remainder}), under Assumptions \ref{ass:eigsandinsts} and \ref{ass:Mhat}, the second term in (\ref{H_conv}) is bounded by
		\begin{equation}
			K \left\Vert \frac{1}{n}Z^\prime \epsilon \right\Vert \left\Vert Z\right\Vert \left\Vert R\right\Vert =O_p\left(\frac{\sqrt{n}}{p^{\nu-1/2}} \right),
		\end{equation}
		which is negligible compared to the first term in (\ref{H_conv}) since $n/p^{(\nu+1/2)}=o(1)$, under Assumptions \ref{ass:eigsandinsts} and  \ref{ass:Mhat}.
		Similarly, the third term in (\ref{H_conv}) is $O_p(n p^{-2\nu})$, which is negligible compared to the first term since $n^{3/2}/{p^{2\nu+ 1}}=o(1)$ as $n\rightarrow \infty$, under Assumptions \ref{ass:eigsandinsts} and \ref{ass:Mhat}. 
		We conclude that
		\begin{equation}
			\left\Vert \hat{H} - H\right\Vert = O_p\left(p/\sqrt{n}\right).
		\end{equation}
		By Assumption \ref{ass:Mhat}, the last term in (\ref{theorem2_1}) is thus $O_p(p^2/n^{3/2})$, given $\Vert d\Vert=O_p(\sqrt{p/n})$ and  $\left\Vert\hat{d}\right\Vert=O_p(\sqrt{p/n})$. Hence, the second term in (\ref{theorem2_1}) is $o_p(\sqrt{p}/n)$ as long as $p^3/n=o(1)$, concluding the proof.
	\end{proof}

	\begin{proof}[Proof of Theorem \ref{theorem:nulldist}]
		By Theorem \ref{theorem:approx}, it suffices to show that
		\[
		\frac{n d^\prime H^{-1} d - p}{\sqrt{2p}}\overset{d}\rightarrow N(0,1).
		\]	
		Under the null hypothesis, the above CLT is for a SAR model with an increasing number of instruments. By Remark 3 of \cite{Gupta2018c}, the claim follows by a trivial modification of the arguments in the proof of Theorem 3.3 therein to allow for heteroskedastic innovations (see e.g. Lemma A.7 in the Supporting Information of \cite{Korolev2026} and the justification of claim (D.3) in \cite{Robinson2008a}).
	\end{proof}
	
	\begin{proof}[Proof of Theorem \ref{theorem:consistency}]

		\noindent Let $\gamma=(\alpha_1,\ldots\alpha_p, \lambda, \beta^\prime)^\prime $, partitioned as $\gamma=(\alpha^\prime, \theta^\prime)^\prime$. Corresponding to $\tilde{d}= \partial \mathcal{Q}/\partial \gamma$ defined in (\ref{d}) under $\mathcal{H}_{0A}$,  we now define the unconstrained gradient vector, $(p+k+1)\times 1$, $\tilde{d}_{U}$ as 
		\begin{equation}\label{full_d_H0}
			\tilde{d}_{U}(\alpha, \lambda, \beta, y)=   -\frac{2}{n} U^\prime   P_Z \left( S_p(\lambda, \alpha, y)-X\beta\right),
		\end{equation}
		where $\tilde{d}_U(0_{p\times 1}, \lambda, \beta, y)= \tilde{d}$ defined in (\ref{d}).
		%partitioned as $\tilde{d}_{U}(\alpha, \lambda, \beta)$ as
		%\begin{equation}\label{partition_d}
		%\tilde{d}_{U}(\alpha, \lambda, \beta)= (\tilde{d}_{U1}(\alpha, \lambda, \beta)^\prime, \mathbf{\tilde{d}}_{U2}(\alpha, \lambda, \beta)^\prime))^\prime, 
		%\end{equation} 
		%with the first and second components of the partitioned vector having, respectively, $p\times 1$ and $(k+1)\times 1$ elements.
		
		We partition $, 
		\hat{J}=n^{-1}U $ as $\hat{J}= (\hat{\Xi}, \hat{N})$, where $\hat{\Xi}$ and $\hat{N}$ are $m\times p$ and $m\times (k+1)$, respectively, with a similar partition for its expected value $J= (\Xi, N)$. Also, we define the $(p+k+1) \times (p+k+1)$ matrix $\hat{D}= \partial ^2\mathcal{Q}/\partial \gamma\partial \gamma^\prime$,  such that the first $p\times p$ block is given by
		\begin{equation}
			\hat{D}_{11}= \frac{2}{n} \begin{pmatrix} \Upsilon_{1}(y)^\prime \\ \ldots\\ \Upsilon_p(y)^\prime \end{pmatrix}  P_Z \begin{pmatrix} \Upsilon_{1}(y)^\prime \\ \ldots\\ \Upsilon_p(y)^\prime \end{pmatrix}^\prime = 2 \hat{\Xi}^\prime \hat{M}^{-1} \hat{\Xi}
		\end{equation}
		the block 1-2 (or the transposed of 2-1 block) is the $p\times (k+1)$ matrix
		\begin{equation}
			\hat{D}_{12} =\hat{D}_{21}^{\prime}= \frac{2}{n} \begin{pmatrix} \Upsilon_{1}(y)^\prime \\ \ldots\\ \Upsilon_p(y)^\prime \end{pmatrix}  P_Z \begin{pmatrix} Wy & X  \end{pmatrix} = 2 \hat{\Xi}^{\prime}\hat{M}^{-1} \hat{N}
		\end{equation}
		and the 2-2 block is the $(k + 1) \times (k+1)$ matrix
		\begin{equation}
			\hat{D}_{22} = \frac{2}{n} \begin{pmatrix} X^\prime \\  y^\prime W^\prime \end{pmatrix}  P_Z \begin{pmatrix} Wy & X  \end{pmatrix} = 2 \hat{N}^\prime \hat{M}^{-1} \hat{N}.
		\end{equation}
		Under Assumption \ref{ass:eigsandinsts}, $\Vert \hat{D} \Vert =O_p(1)$ and $\liminf_{n\rightarrow\infty}\underline{{eig}}\left(\hat{D} \right) >0$ with inverse defined and partitioned in the usual way. Also, $\hat{D}$ does not depend on any unknowns. In line we our previous notation, we also define the corresponding limit quantities as $D_{11}= 2 \Xi^{\prime} M^{-1} \Xi$, $D_{12}= D_{21}^\prime= \Xi^\prime M^{-1} N$ and $D_{22}= 2 N^\prime M^{-1} N$.
		
		From standard algebra, by the mean value theorem (MVT), given $\hat{d}$ in (\ref{dhat}),
		\begin{align}
			& \hat{d}_p= \left .\frac{\partial \mathcal{Q}}{\partial \alpha^\prime}\right\vert _{(0_{1 \times p}, \hat{\theta}^\prime)^\prime} = \left. \frac{\partial \mathcal{Q}}{\partial \alpha^\prime}\right\vert_{(0_{1\times p}, \theta_0^\prime)^\prime}+ \hat{D}_{12}(\hat{\theta}- \theta_0),\notag \\
			& 0= \left .\frac{\partial \mathcal{Q}}{\partial \theta^\prime}\right\vert _{(0_{1 \times p}, \hat{\theta}^\prime)^\prime} = \left. \frac{\partial \mathcal{Q}}{\partial \theta^\prime}\right\vert_{(0_{1\times p}, \theta_0^\prime)^\prime}+ \hat{D}_{22}(\hat{\theta}- \theta_0),
		\end{align}
where we now and subsequently denote true parameter values by $0$ subscripts. Thus,
		\begin{align}
			\hat{d}_p= \left(I_p ; \  -\hat{D}_{12}\hat{D}_{22}^{-1}\right)\left.\begin{pmatrix}\frac{\partial \mathcal{Q}}{\partial \alpha^\prime} \\ \frac{\partial \mathcal{Q}}{\partial \theta^\prime} \end{pmatrix}\right\vert_{(0_{1\times p}, \theta_0^\prime)^\prime} =&  \left(I_p ; \  -\hat{D}_{12}\hat{D}_{22}^{-1}\right) \tilde{d}_{U}(0_{p\times 1}, \lambda_0, \beta_0) \notag \\
			=& \left(I_p ; \  -\hat{D}_{12}\hat{D}_{22}^{-1}\right) \tilde{d}(\lambda_0, \beta_0)
		\end{align}
		according to the definition in (\ref{full_d_H0}) and (\ref{d}), and with $I_p$ denoting the $p\times p$ identity matrix. Hence, given $\hat{H}$ in (\ref{Hhat}),
		\begin{equation}\label{equiv_th4}
			n\hat{d}_p^\prime \hat{H}^{11} \hat{d}_p = n \tilde{d}_U(0_{p\times 1}, \lambda_0, \beta_0)^\prime \hat{\mathcal{V}}\tilde{d}_{U}(0_{p\times 1}, \lambda_0, \beta_0),
		\end{equation}
		with 
		\begin{align}\label{italic_V}
			\hat{\mathcal{V}}=& \begin{pmatrix} I_{p} \\ - \hat{D}_{22}^{-1} \hat{D}_{21}\end{pmatrix} \hat{H}^{11} \left(I_{p}\ ; \ - \hat{D}_{12}\hat{D}_{22}^{-1} \right). 
		\end{align}
		
		Thus, 
		\begin{align}\label{stat1}
			n \hat{d}^\prime \hat{H}^{-1}\hat{d}= n\hat{d}_p^\prime \hat{H}^{11} \hat{d}_p = n \tilde{d}_{U}(0_{p\times 1}, \lambda_0, \beta_0)^\prime \hat{\mathcal{V}} \tilde{d}_{U}(0_{p\times 1}, \lambda_0, \beta_0) 
		\end{align}
		
		However,  under $\mathcal{H}_{1A}$, $\tilde{d}_{U}(0_{p\times 1}, \lambda_0, \beta_0)$  is no longer evaluated at the true parameter value as $\alpha_0 \neq 0$.  By MVT around $\alpha_0$, we can write
		\begin{equation}\label{d_H1}
			\tilde{d}_U(0_{p\times 1}, \lambda_0, \beta_0)= \tilde{d}_U(\alpha_0,\lambda_0, \beta_0)  - \frac{\partial \tilde{d}_U(\bar{\alpha}, \lambda_0, \beta_0)}{\partial \alpha} \alpha_0 \equiv \tilde{d}_U(\alpha_0,\lambda_0, \beta_0)   - \tau,
		\end{equation}
		with $\bar{\alpha}$ being  intermediate point such that $\Vert \bar{\alpha}-\alpha_0 \Vert \leq \Vert \alpha_0\Vert$ and $\tau$ being the $(p+k+1) \times 1$ vector defined as
		\begin{equation}
			\tau=-	\frac{\partial \tilde{d}_U(\bar{\alpha}, \lambda_0, \beta_0)}{\partial \alpha} \alpha_0= \frac{2}{n}U^\prime  P_Z \begin{pmatrix} \Upsilon_{1}(y) & \ldots& \Upsilon_p(y) \end{pmatrix} \alpha_0 =  \hat{J}^\prime \hat{M}^{-1}   \hat{\Xi} \alpha_0.
		\end{equation}
		%with corresponding expressions for $\tilde{\mathbf{d}}_{U1}(0_{p\times 1}, \lambda_0, \beta_0)$ and $\tilde{\mathbf{d}}_{U2}(0_{p\times 1}, \lambda_0, \beta_0)$, where
		%\begin{equation}
		%\tau_1=	\frac{\partial \tilde{\mathbf{d}}_{U1}(\bar{\alpha}, \lambda_0, \beta_0)}{\partial \alpha} \alpha_0=   \hat{\Xi}^\prime \hat{M}^{-1}   \hat{\Xi} \alpha_0 \ \ \ \  \tau_2=	\frac{\partial \tilde{\mathbf{d}}_{U2}(\bar{\alpha}, \lambda_0, \beta_0)}{\partial \alpha} \alpha_0=   \hat{N}^\prime \hat{M}^{-1}   \hat{\Xi} \alpha_0. 
		%\end{equation}
		Similarly to (\ref{theorem1_1}) and (\ref{d_remainder}), 
		\begin{align}
			\Vert \tilde{d}_U(\alpha_0,\lambda_0, \beta_0)\Vert \leq & K \Vert\hat{J} \Vert \Vert \hat{M}^{-1}\Vert \left \Vert \frac {1}{n} Z^\prime \epsilon \right\Vert +  K \Vert\hat{J} \Vert \Vert \hat{M}^{-1}\Vert \left \Vert \frac {1}{n} Z^\prime R \right\Vert \notag \\= &O_p\left(\max\left(\sqrt{\frac{p}{n}}, p^{-\nu} \right)\right) = O_p \left(\sqrt{\frac{p}{n}}\right)
		\end{align}
		for $\nu$ satisfying $\sqrt{n}/p^{\nu+1/2}=o(1)$, which holds under Assumption \ref{ass:eigsandinsts}, and $\Vert \tau \Vert=O_p(1)$ and non-zero, since $\alpha_0\neq 0$. 
		%Similar rates apply to $\Vert \tilde{\mathbf{d}}_{U1}(\alpha_0,\lambda_0, \beta_0)\Vert$,  $\Vert \tilde{\mathbf{d}}_{U2}(\alpha_0,\lambda_0, \beta_0)\Vert$, $\Vert \tau_1 \Vert$ and $\Vert \tau_2\Vert$.

		We furthermore define the unconstrained version of $\hat{H}$ evaluated at generic parameters' value as
		\begin{equation}
			\tilde{H}_U(\alpha, \lambda, \beta)= 4 \hat{J}^\prime \hat{M}^{-1} \tilde{\Omega}_U(\alpha, \lambda, \beta) \hat{M}^{-1} \hat{J},  
		\end{equation}
		partitioned in the usual way, where $\tilde{\Omega}_U$ is defined according to (\ref{epsilonU_def}). We also define its limit quantity $H_U(\alpha_0, \lambda_0, \beta_0)= 4 J^\prime M^{-1} \Omega M^{-1} J$, where, as previously defined, $\Omega= n^{-1}\mathbb{E}(Z^\prime \Sigma Z)$ and $\Sigma$ the $n\times n$ diagonal matrix with diagonal given by $\sigma_i^2, i=1,\ldots,n$.
		Similarly to what deduced in (\ref{omega1}) and (\ref{omega2}), under Assumptions \ref{ass:eigsandinsts}-\ref{ass:power}, $\Vert \tilde{H}_U (\alpha, \lambda, \beta)\Vert = O_p(1)$, uniformly in $(\alpha, \lambda, \beta)$ and $\liminf_{n\rightarrow\infty}\underline{\textit{eig}}(\tilde{H}_U(\alpha, \lambda, \beta))> c>0$, uniformly in $(\alpha, \lambda, \beta)$ and almost surely. 
		
		Clearly, $\hat{H}= \tilde{H}_U(0, \hat{\lambda}, \hat{\beta})$. We can apply the MVT to $\hat{H}^{-1}$ around the true parameters' value and obtain
		\begin{align}\label{H_H1}
			\hat{H}^{-1}= &\tilde{H}_U^{-1}(\alpha_0, \lambda_0,\beta_0)+  \sum_{j=1}^p \tilde{H}_U^{-1}(\bar{\alpha}, \bar{\lambda}, \bar{\beta}) \frac{\partial \tilde{H}_U}{\partial{\alpha_j}}\vert_{(\bar{\alpha}, \bar{\lambda}, \bar{\beta})} \tilde{H}_U^{-1}(\bar{\alpha}, \bar{\lambda}, \bar{\beta}) \alpha_{0j} \notag \\
			-& \sum_{t=1}^k \tilde{H}_U^{-1}(\bar{\alpha}, \bar{\lambda}, \bar{\beta}) \frac{\partial \tilde{H}_U}{\partial{\beta_t}}\vert_{(\bar{\alpha}, \bar{\lambda}, \bar{\beta})} \tilde{H}_U^{-1}(\bar{\alpha}, \bar{\lambda}, \bar{\beta}) (\hat{\beta}_t - \beta_{0t}) \notag \\
			- &\tilde{H}_U^{-1}(\bar{\alpha}, \bar{\lambda}, \bar{\beta}) \frac{\partial \tilde{H}_U}{\partial{\lambda}}\vert_{(\bar{\alpha}, \bar{\lambda}, \bar{\beta})} \tilde{H}_U^{-1}(\bar{\alpha}, \bar{\lambda}, \bar{\beta}) (\hat{\lambda} - \lambda_{0})
			\equiv  \tilde{H}_U^{-1}(\alpha_0, \lambda_0,\beta_0) + T,
		\end{align}
		where $\bar{\alpha}$, $\bar{\beta}$ and $\bar{\lambda}$ are intermediate points such that $\Vert \bar{\alpha}- \alpha_0\Vert \leq \Vert \alpha_0\Vert$, $\Vert \bar{\beta}- \beta_0 \Vert \leq \Vert \hat{\beta}- \beta_0 \Vert$ and $\vert \bar{\lambda}- \lambda_0 \vert \leq \vert \hat{\lambda}- \lambda_0 \vert$.  Under $\mathcal{H}_{0A}$, $\Vert T\Vert = O_p(\sqrt{p/n}).$Under $\mathcal{H}_{1A}$, $\alpha_{0j} \neq 0$ for some $j=1,\ldots,p$ and, since $\hat{\lambda}$ and $\hat{\beta}_{t}$ for $t=1,\ldots,k$ are restricted estimates, $\hat{\lambda} - \lambda_{0} = O_p(1)$ and $\hat{\beta}_t - \beta_{0t}=O(1)$ for some $t=1,\ldots,k$. Thus, under Assumptions \ref{ass:eigsandinsts}-\ref{ass:power}, $\Vert T \Vert =O_p(p)$ and $\liminf_{n\rightarrow\infty}\underline{\textit{eig}}(T)> c>0$. By partitioning $T$ in the usual way, we obtain $\hat{H}^{11} = \tilde{H}_U^{11}(\alpha_0, \lambda_0, \beta_0) + T_{11}$.  Also, let 
		\begin{equation}\label{V_def}
			\tilde{\mathcal{V}}(\alpha_0, \lambda_0, \beta_0)= \begin{pmatrix} I_{p} \\ - \hat{D}_{22}^{-1} \hat{D}_{21}\end{pmatrix} \tilde{H}_U^{11}(\alpha_0, \lambda_0, \beta_0) \left(I_{p}\ ; \ - \hat{D}_{12}\hat{D}_{22}^{-1} \right)
		\end{equation}
		and
		\begin{equation}\label{W_limit}
			\tilde{\mathcal{W}}= \begin{pmatrix} I_{p} \\ - \hat{D}_{22}^{-1} \hat{D}_{21}\end{pmatrix} T_{11} \left(I_{p}\ ; \ - \hat{D}_{12}\hat{D}_{22}^{-1} \right).
		\end{equation}
		
		%Trivially, under  $\mathcal{H}_{0p}$, by defining $d_U = -2/n J^\prime M^{-1} Z^\prime \epsilon$, $d_p= -2/n \Xi^\prime M^{-1/2} \mathcal{M}_{MN} M^{-1/2} Z^\prime \epsilon$, and observing that, after straightforward algebra, $(I_p; -D_{12} D_{22}^{-1}) d_U = d_p$, we have
		%\begin{equation}
		%\frac{n \tilde{d}_{U}(0_{p\times 1}, \lambda_0, \beta_0)^\prime \hat{\mathcal{V}}\tilde{d}_{U}(0_{p\times 1}, \lambda_0, \beta_0) -p}{2^{1/2} p^{1/2}} \overset{d}{\rightarrow} \mathcal{N}(0,1).
		%\end{equation}
		
		\noindent From (\ref{d_H1}) and (\ref{H_H1}), (\ref{stat1}) becomes
		\begin{align}
			n\hat{d}_p^\prime \hat{H}^{11} \hat{d}_p =& n \tilde{d}_{U}(\alpha_0, \lambda_0, \beta_0)^\prime \tilde{\mathcal{V}}(\alpha_0, \lambda_0, \beta_0) \tilde{d}_{U}(\alpha_0, \lambda_0, \beta_0) + 2n \tau^\prime \tilde{\mathcal{V}}(\alpha_0, \lambda_0, \beta_0) \tilde{d}_{U}(\alpha_0, \lambda_0, \beta_0) \notag \\
			+& n \tau^\prime \tilde{\mathcal{V}}(\alpha_0, \lambda_0, \beta_0) \tau + n \tilde{d}_{U}(\alpha_0, \lambda_0, \beta_0)^\prime \tilde{\mathcal{W}} \tilde{d}_{U}(\alpha_0, \lambda_0, \beta_0) \notag\\
			&+ 2n \tau^\prime \tilde{\mathcal{W}} \tilde{d}_{U}(\alpha_0, \lambda_0, \beta_0) 
			+ n \tau^\prime \tilde{\mathcal{W}} \tau,
		\end{align}
		and thus
		\begin{align}\label{statistic_h1}
			\frac{n\hat{d}_p^\prime \hat{H}^{11} \hat{d}_p -p}{(2p)^{1/2}}&= \frac{n \tilde{d}_{U}(\alpha_0, \lambda_0, \beta_0)^\prime \tilde{\mathcal{V}}(\alpha_0, \lambda_0, \beta_0) \tilde{d}_{U}(\alpha_0, \lambda_0, \beta_0)-p}{(2p)^{1/2}}\notag\\
			& + \frac{\sqrt{2}n}{\sqrt{p}}\tau^\prime \tilde{\mathcal{V}}(\alpha_0, \lambda_0, \beta_0) \tilde{d}_{U}(\alpha_0, \lambda_0, \beta_0) \notag \\
			&+  \frac{n}{\sqrt{2p}}\tau^\prime \tilde{\mathcal{V}}(\alpha_0, \lambda_0, \beta_0) \tau + \frac{n}{\sqrt{2p}}\tilde{d}_{U}(\alpha_0, \lambda_0, \beta_0)^\prime \tilde{\mathcal{W}} \tilde{d}_{U}(\alpha_0, \lambda_0, \beta_0)\notag\\ 
			&+ \frac{\sqrt{2} n}{\sqrt{p}}\tau^\prime \tilde{\mathcal{W}} \tilde{d}_{U}(\alpha_0, \lambda_0, \beta_0) + \frac{ n}{\sqrt{2p}} \tau^\prime \tilde{\mathcal{W}} \tau 
		\end{align}
		By a similar argument adopted in the proof of Theorem \ref{theorem:dhatd}, we can show $\Vert \tilde{d}_U(\alpha_0, \lambda_0, \beta_0) - d_U \Vert = O_p(p^{3/2}/n)$, with $d_U= -2/n J^\prime M^{-1} Z^\prime \epsilon$ and $d_p= (I_p; -D_{12} D_{22}^{-1}) d_U $. Also, we can show
		\begin{equation}\label{H_U convergence}
			\Vert \tilde{H}_U(\alpha_0, \lambda_0, \beta_0) - H_U \Vert = O_p\left(\frac{p}{\sqrt{n}}\right),
		\end{equation}
		such that, under Assumptions \ref{ass:eigsandinsts}-\ref{ass:power},  $\Vert \tilde{H}^{11}_U(\alpha_0, \lambda_0, \beta_0) - H_U^{11} \Vert = O_p\left(p/\sqrt{n}\right)$.
		We show the claim in (\ref{H_U convergence}) by routine arguments as in (\ref{H_conv_first_1}) and (\ref{omega_terms}), after observing that 
		$\tilde{H}_U(\alpha_0, \lambda_0, \beta_0) = 4 \hat{J}^{\prime} \hat{M}^{-1} \tilde{\Omega}_R \hat{M}^{-1} \hat{J}$, with $\tilde{\Omega}_R = \tilde{\Omega} + \sum_{i=1}^{p} z_i z_i^\prime R_i^2 /n$, and 
		\begin{align}
			&\Vert \tilde{\Omega}_R - \Omega \Vert \leq \Vert \tilde{\Omega} - \Omega \Vert + \left\Vert \frac{\sum_{i=1}^p z_i z_i^\prime R_i^{2}}{n}\right\Vert
			=  O_p\left(\frac{p}{\sqrt{n}}\right) + \underset{1\leq i\leq n}{\sup} R_i^2 \ \Vert \hat{M} \Vert  \notag \\= & O_p\left(\frac{p}{\sqrt{n}}\right) + O_p(p^{-2\nu})= O_p\left(\frac{p}{\sqrt{n}}\right),
		\end{align}
		where the last equality follows for $\nu$ satisfying $\sqrt{n}/p^{2\nu+1}=o(1)$, which holds under Assumption \ref{ass:eigsandinsts}.

		After showing, similarly to steps in the proof of Theorem \ref{theorem:approx}, that
		\begin{equation}
			\tilde{d}_{U}(\alpha_0, \lambda_0, \beta_0)^\prime \tilde{\mathcal{V}}(\alpha_0, \lambda_0, \beta_0) \tilde{d}_{U}(\alpha_0, \lambda_0, \beta_0) -  d_{p}^\prime {H}_U^{11}d_{p}=o_p\left(\frac{\sqrt{p}}{n}\right), 
		\end{equation}
		we conclude that the first term in (\ref{statistic_h1}) is $O_p(1)$, as shown in Theorem  \ref{theorem:nulldist}. 
		By standard norm inequalities, the second term in (\ref{statistic_h1}) is $O_p(\sqrt{n})$, the third is $O_p(n/\sqrt{p})$, the fourth is $O_p(p^{3/2})$,  the fifth is $O_p(p\sqrt{n})$ and  the sixth is $O_p(n\sqrt{p})$. The last term dominates the former five ones and thus, under $\mathcal{H}_{1A}$, for all $\eta>0$, $\mathbb{P}\left(|\mathcal{T}|^{-1} \leq \eta/n \sqrt{p} \right) \rightarrow 1$ as $n\rightarrow \infty $ and hence consistency of $\mathcal{T}$ follows.
	\end{proof}
	
	\begin{proof}[Proof of Theorem \ref{theorem:local}] 
		
		\noindent Similarly to the proof of Theorem \ref{theorem:consistency}, we write
		\begin{align}
			\tilde{d}_U(0_{p\times 1}, \lambda_0, \beta_0)&= \tilde{d}_U(\alpha_{0n},\lambda_0, \beta_0)  - \frac{\partial \tilde{d}_U(\bar{\alpha}_n, \lambda_0, \beta_0)}{\partial \alpha} \alpha_{0n}\notag\\
			&= \tilde{d}_U(\alpha_{0n},\lambda_0, \beta_0)  - \frac{\partial \tilde{d}_U(\bar{\alpha}_n, \lambda_0, \beta_0)}{\partial \alpha} \frac{p^{1/4}}{\sqrt{n}}\delta, \label{d_Hl}
		\end{align}
		where we denote by $\alpha_{0n}$ the value of the $p\times 1$ vector $\alpha$ under $\mathcal{H}_{\ell}$, $\delta= (\delta_1,\ldots., \delta_p)^\prime$, $\Vert \bar{\alpha}_n - \alpha_{0n}\Vert \leq \Vert \alpha_{0n}\Vert \sim p^{1/4}/n$ since $\Vert \delta \Vert =1$ and 
		\begin{equation}
			\tau_n= \frac{\partial \tilde{d}_U(\bar{\alpha}, \lambda_0, \beta_0)}{\partial \alpha} \alpha_{0n}=   \hat{J}^\prime \hat{M}^{-1}   \hat{\Xi} \frac{p^{1/4}}{\sqrt{n}}\delta,
		\end{equation}
		such that $\Vert \tau_n\Vert = O(p^{1/4}/\sqrt{n})$.
		
		Also, by consistency of  $\left(\hat\lambda,\hat\beta'\right)'$, under $\mathcal{H}_{l}$ under Assumptions \ref{ass:regressors} and \ref{ass:eigsandinsts}-\ref{ass:power}, 
		\begin{align}
			\hat{H}=&  \tilde{H}_U(\alpha_{0n}, \lambda_0,\beta_0)-  \sum_{j=1}^p  \frac{\partial \tilde{H}_U}{\partial{\alpha_j}}\vert_{(\bar{\alpha}, \bar{\lambda}, \bar{\beta})} \alpha_{0n,j} + T_{1n},
		\end{align}
		with $\Vert T_{1n}\Vert = O_p(\sqrt{p/n})$. 
		Let
		\begin{equation}
			T_{2n}= \sum_{j=1}^p \frac{\partial \tilde{H}_U}{\partial{\alpha_j}}\vert_{(\bar{\alpha}, \bar{\lambda}, \bar{\beta})}  \alpha_{0n,j} = \frac{p^{1/4}}{\sqrt{n}} \sum_{j=1}^p   \frac{\partial \tilde{H}_U}{\partial{\alpha_j}}\vert_{(\bar{\alpha}, \bar{\lambda}, \bar{\beta})}   \delta_{j}. 
		\end{equation}
		By standard norm inequalities, under Assumptions \ref{ass:eigsandinsts}-\ref{ass:power}, we write
		\begin{equation}
			\Vert T_{2n}\Vert \leq \frac{p^{1/4}}{n^{1/2}}\underset{1\leq j\leq p}{\sup} \left \Vert \frac{\partial \tilde{H}_U}{\partial{\alpha_j}}\vert_{(\bar{\alpha}, \bar{\lambda}, \bar{\beta})}\right \Vert \sum_{j=1}^{p}|\delta_j| \leq  \frac{p^{3/4}}{n^{1/2}}\underset{1\leq j\leq p}{\sup} \left \Vert \frac{\partial \tilde{H}_U}{\partial{\alpha_j}}\vert_{(\bar{\alpha}, \bar{\lambda}, \bar{\beta})}\right \Vert \Vert \delta\Vert ^2 = O_p\left(\frac{p^{3/4}}{\sqrt{n}}\right).
		\end{equation}
		Thus, 
		\begin{equation}\label{Tn_rate}
			\Vert T_n\Vert \equiv \Vert T_{1n} + T_{2n}\Vert \leq \Vert T_{1n}\Vert + \Vert T_{2n}\Vert =O_p(p^{3/4}/\sqrt{n}).
		\end{equation}
		Under $\mathcal{H}_l$, we also obtain
		\begin{align}\label{HhattoHU}
			&\Vert \hat{H} - H_U\Vert \leq \Vert \hat{H} - \tilde{H}_U(\alpha_{0n}, \lambda_0, \beta_0) \Vert + \Vert \tilde{H}_U(\alpha_{0n}, \lambda_0, \beta_0) - H_U \Vert \notag \\&= O_p\left(\frac{p^{3/4}}{n}\right)+ O_p\left(\frac{p}{\sqrt{n}}\right)=O_p\left(\frac{p}{\sqrt{n}}\right),
		\end{align}
		where the first equality follows from (\ref{Tn_rate}) and (\ref{H_U convergence}). Hence,
		\begin{align}\label{statistic_hl}
			\frac{n\hat{d}_p^\prime \hat{H}^{11} \hat{d}_p -p}{(2p)^{1/2}}&= \frac{n \tilde{d}_{U}(\alpha_{0n}, \lambda_0, \beta_0)^\prime \hat{\mathcal{V}}_n \tilde{d}_{U}(\alpha_{0n}, \lambda_0, \beta_0)-p}{(2p)^{1/2}} \notag\\
			&+  \frac{\sqrt{2}n}{\sqrt{p}}\tau_n^\prime \hat{\mathcal{V}}_n \tilde{d}_{U}(\alpha_{0n}, \lambda_0, \beta_0) 
			+  \frac{n}{\sqrt{2p}}\tau_n^\prime \hat{\mathcal{V}}_n \tau_n .
		\end{align}
		After showing, by standard arguments and using (\ref{HhattoHU}), 
		\begin{equation}
			\tilde{d}_U(\alpha_{0n}, \lambda_0, \beta_0)^\prime \hat{\mathcal{V}} \tilde{d}_U(\alpha_{0n}, \lambda_0, \beta_0) - d_p^\prime H_U^{11}d_p = o_p\left(\frac{\sqrt{p}}{n}\right),
		\end{equation}
		we conclude that the first term in (\ref{statistic_hl}) converges to $\mathcal{N}(0,1)$ as $n\rightarrow \infty$.
		The third term in (\ref{statistic_hl}) has a finite limit, since
		\begin{equation}
			\frac{n}{\sqrt{2p}}\tau_n^\prime \hat{\mathcal{V}}_n(\alpha_{0n}, \lambda_0, \beta_0) \tau_n - \delta^\prime \Xi^\prime M^{-1} J \mathcal{V} J^\prime M^{-1} \Xi \delta = o_p(1), 
		\end{equation}
		where 
		\begin{equation}
			\mathcal{V}= \begin{pmatrix} I_{p} \\ - D_{22}^{-1} D_{21}\end{pmatrix} H_U^{11}\left(I_{p}\ ; \ - D_{12}D_{22}^{-1} \right)
		\end{equation}
		and $\delta^\prime \Xi^\prime M^{-1} J \mathcal{V} J^\prime M^{-1} \Xi \delta \neq 0$ unless $\delta=0$.
		The second term in (\ref{statistic_hl}) can be written as
		\begin{equation}\label{second_local}
			\frac{\sqrt{2}n}{\sqrt{p}}\tau_n^\prime \hat{\mathcal{V}}_n \tilde{d}_{U} = \frac{\sqrt{2}n}{\sqrt{p}}\tau_n^\prime (\hat{\mathcal{V}}_n - \mathcal{V}) \tilde{d}_{U} +  \frac{\sqrt{2}n}{\sqrt{p}}\tau_n^\prime \mathcal{V}(\tilde{d}_U -d_U) + \frac{\sqrt{2}n}{\sqrt{p}} \tau_n^\prime \mathcal{V} d_U.
		\end{equation}
		By routine arguments, the first two terms of last displayed expression are, respectively, $O_p(p^{5/4}/n^{1/2})=o_p(1)$ and $O_p(p^{3/4}/n^{1/2})$, and are thus $o_p(1)$ as long as $p^3/n=o(1)$. The third term in (\ref{second_local}) is equivalent to
		\begin{align}\label{second_local_1}
			&\frac{\sqrt{2}n^{1/2}}{p^{1/4}} \delta^\prime (\hat{\Xi}^\prime - \Xi^\prime) \hat{M}^{-1}\hat{J}\mathcal{V}d_U + \frac{\sqrt{2}n^{1/2}}{p^{1/4}} \delta^\prime \Xi^\prime (\hat{M}^{-1}- M^{-1})\hat{J}\mathcal{V}d_U \notag \\
			& + \frac{\sqrt{2}n^{1/2}}{p^{1/4}} \delta^\prime \Xi^\prime M^{-1}(\hat{J}- J) \mathcal{V}d_U 
			+ \frac{\sqrt{2}n^{1/2}}{p^{1/4}} \delta^\prime \Xi^\prime M^{-1}J \mathcal{V}d_U.
		\end{align}
By routine arguments adopted throughout the proofs, the first three terms in (\ref{second_local_1}) are $O_p(p^{5/4}/n^{1/2})=o_p(1)$. The last term in (\ref{second_local_1}) has mean zero and variance bounded by
		\begin{align}
			& K \frac{1}{p^{1/2}}\delta^\prime \Xi^\prime M^{-1} J \mathcal{V}J^\prime M^{-1} \Omega M^{-1} J \mathcal{V}J^\prime M^{-1}\Xi \delta  \leq  K \frac{1}{p^{1/2}} \Vert\delta\Vert ^2 \Vert \Xi\Vert ^2 \Vert J\Vert^4 \Vert M^{-1} \Vert^4 \Vert \mathcal{V}\Vert^2 \Vert \Omega \Vert \notag\\
			&= O\left(\frac{1}{p^{1/2}}\right),
		\end{align}
		such that the last term in (\ref{second_local_1}) is $O_p(1/p^{1/4})$. Hence, the second term in (\ref{second_local}) is $O_p(1/p^{1/4})=o_p(1)$ since $p\rightarrow \infty$,  concluding the claim as
		\begin{equation}
			\frac{n\hat{d}_p^\prime \hat{H}^{11} \hat{d}_p -p}{\sqrt{2p}} \overset{d}{\rightarrow} \mathcal{N}(\varrho, 1),
		\end{equation}
		with $\varrho=\delta^\prime \Xi^\prime M^{-1} J \mathcal{V} J^\prime M^{-1} \Xi \delta  $.
	\end{proof}
	
	%	\begin{proof}[Proof of Proposition \ref{prop:NEDsuffrate}]
	%		Let $\upsilon_{jl}(y)$ denote the elements of $\Upsilon_l(y)$ as defined in (\ref{NEDUpsilon}). Then, because $\mathcal{S}$ is bounded, Proposition 2 of \cite{Jenish2012} implies that $\upsilon_{jl}(y)$ inherits the NED properties of $w_i'y$ in Assumption \ref{ass:NEDzwy}. Furthermore, the $L_4$-NED property of $z_{il}$ and $w_i'y$ implies that $t_{il_1l_2}=z_{il_1}\upsilon_{il_2}$ are uniformly $L_2$-NED with NED coefficients $\phi(m)$ and Lemma A.3 of \cite{Jenish2012} implies that $var\left(\sum_{i=1}^n\left(t_{il_1l_2}-Et_{il_1l_2}\right)\right)=O(n)$. Then
	%		\begin{eqnarray*}
	%			E\left\Vert\hat J-J\right\Vert^2&\leq &n^{-2}\sum_{l_1=1}^m \sum_{l_2=1}^p E\left(\sum_{i=1}^n\left(t_{il_1l_2}-Et_{il_1l_2}\right)\right)^2\\
	%			&=&	n^{-2}\sum_{l_1=1}^m \sum_{l_2=1}^p var\left(\sum_{i=1}^n\left(t_{il_1l_2}-Et_{il_1l_2}\right)\right)
	%			= O(p^2/n),
	%		\end{eqnarray*}
	%		because $m\sim p$. The conclusion now follows by Markov's inequality.
	%	\end{proof}	
	
	\begin{proof}[Proof of Theorem \ref{thm:quad_moms}]
		
		The proof closely follows that of Theorem 4.1 in \cite{Yang2025}. By the first order condition and Taylor expansion, we have
		\begin{equation*}
			\mathring{\theta}- \theta_{0} = -\left[Q_{1}^{\prime}(\mathring{\gamma})\check{\Psi}^{-1}Q_{1}(\bar{\gamma})\right]^{-1}Q_{1}^{\prime}(\mathring{\gamma})\check{\Psi}^{-1}q(\gamma_0),
		\end{equation*}
		\noindent where $Q_{1}(\mathring{\gamma}) = -n^{-1}\left(L_{1}^s\epsilon(\mathring{\gamma}),\ldots,L_{\ell}^s\epsilon(\mathring{\gamma}),Z\right)'\mathbb{X}$ with  $\epsilon(\mathring{\gamma}) = y-  \mathring{\lambda}Wy-X\mathring\beta$, and $\bar{\gamma}$ lies between $\mathring{\gamma}$ and the true parameter $\gamma_0 = (0_p',\lambda_0,{\beta_0}^{\prime})'=(0_p',\theta_0')'$. Then under the null,
		\begin{align}
			q(\mathring{\gamma}) &= q(\gamma_0) + Q_{1}(\bar{\gamma})(\mathring{\theta}- \theta_{0})\notag\\
			&= \left[I- Q_{1}(\mathring{\gamma})\left(Q_{1}^{\prime}(\mathring{\gamma})\check{\Psi}^{-1}Q_{1}(\mathring{\gamma})\right)^{-1}Q_{1}^{\prime}(\mathring{\gamma})\check{\Psi}^{-1}\right]q(\gamma_0). 
		\end{align}
		Denote $\bar{\mathcal{R}}_{1} = I- \check{\Psi}^{-1/2}Q_{1}(\mathring{\gamma})\left(Q_{1}'(\mathring{\gamma})\check{\Psi}^{-1}Q_{1}(\mathring{\gamma})\right)^{-1}Q_{1}'(\mathring{\gamma})\check{\Psi}^{-1/2}$, where is $\check{\Psi}^{1/2}$ is the square root of $\check{\Psi}$.  Then $\check{\Psi}^{-1/2}q(\mathring{\gamma}) =\bar{\mathcal{R}}_{1}\check{\Psi}^{-1/2}q(\gamma_0)$. Note that
		\begin{equation*}
			\mathcal{S} = \frac{nq^{\prime}(\gamma_0)\check{\Psi}^{-1/2}\bar{\mathcal{R}}_{1}^{\prime}\mathcal{P}(\check{\Psi}^{-1/2}Q(\mathring{\gamma}))\bar{\mathcal{R}}_{1}\check{\Psi}^{-1/2}q(\gamma_0)-p}{\sqrt{2p}} = \frac{\mathcal{S}_1-p}{\sqrt{2p}} + \frac{\mathcal{S}_2}{\sqrt{2p}},
		\end{equation*}
		where
		\begin{align*}
			\mathcal{S}_1 &= nq^{\prime}(\gamma_0)\Psi^{-1/2}\mathcal{R}_{1}\mathcal{P}(\Psi^{-1/2}D)\mathcal{R}_{1}\Psi^{-1/2}q(\gamma_0),\\
			\mathcal{S}_2 &= n\big(\|\mathcal{P}(\check{\Psi}^{-1/2}\bar{Q})\bar{\mathcal{R}}_{1}\check{\Psi}^{-1/2}q(\gamma_0)\|^2 - \|\mathcal{P}(\Psi^{-1/2}D)\mathcal{R}_{1}\Psi^{-1/2}q(\gamma_0)\|^2\big),
		\end{align*}
		with $\bar{Q} \equiv Q(\mathring{\gamma})$ and $\mathcal{R}_{1} = I- \mathcal{P}\big(\Psi^{-1/2}D_{1}\big)$, $D_1=n^{-1}\mathbb{E}\left(\left(L_1^s\epsilon,\ldots,L_\ell^s\epsilon, Z\right)'\mathbb{X}\right)$. We first show that $\mathcal{S}_2 = o_p(p^{1/2})$. As in \cite{Yang2025}, we can show that $\|Q(\bar{\gamma}) - D\| \leq \|Q(\bar{\gamma}) - Q(\gamma_0)\| + \|Q(\gamma_0) - D\| = O_p(p/\sqrt{n})$ and $\|Q_{1}(\bar{\gamma}) - D_{1}\| = O_p(p^{1/2}n^{-1/2})$ for any $\bar{\gamma}$ lying between $\mathring{\gamma}$ and $\gamma_0$. Then
		\begin{align*}
			&\|\check{\Psi}^{-1/2}\bar{\mathcal{R}}_{1}^{\prime}\mathcal{P}(\check{\Psi}^{-1/2}\bar{Q})\bar{\mathcal{R}}_{1}\check{\Psi}^{-1/2} - \Psi^{-1/2}\mathcal{R}_{1}\mathcal{P}(\check{\Psi}^{-1/2}D)\mathcal{R}_{1}\Psi^{-1/2}\|\\
			&= O_p\!\left(\frac{p}{\sqrt{n}} + p^{1/2-\nu}\right).
		\end{align*}
		Since $\|q(\gamma_0)\| = O_p(\sqrt{p/n})$, we have $\mathcal{S}_2 = O_p\big(p^2/\sqrt{n} +p^{3/2-\nu}\big) = o_p(p^{1/2})$ under Assumption \ref{ass:Psicheck}. $\Psi^{-1/2}D_{1}$ being a submatrix of $\Psi^{-1/2}D$ implies $D_{1}'\Psi^{-1}D(D'\Psi^{-1}D)^{-1}D'\Psi^{-1/2} = D_{1}'\Psi^{-1/2}$, which in turn implies that $\mathcal{R}_{1}'\mathcal{P}(\Psi^{-1/2}D)\mathcal{R}_{1} = \mathcal{P}(\Psi^{-1/2}D) - \mathcal{P}(\Psi^{-1/2}D_{1})$. Hence,
		\begin{equation*}
			\mathcal{S}_1 = nq^{\prime}(\gamma_0)\Psi^{-1/2}\big[\mathcal{P}(\Psi^{-1/2}D) - \mathcal{P}(\Psi^{-1/2}D_{1})\big]\Psi^{-1/2}q(\gamma_0).
		\end{equation*}
		Note that $\mathcal{P}(\Psi^{-1/2}D) - \mathcal{P}(\Psi^{-1/2}D_{1})$ is an idempotent matrix with rank $p$. Moreover, under the null, $q(\gamma_0) = (\epsilon'L_{1}\epsilon,\ldots,\epsilon'L_{\ell}\epsilon,\epsilon'Z)'/n$. The theorem now follows from Lemma A.9 of \cite{Yang2025}.
	\end{proof}	
	\clearpage
	\bibliographystyle{chicago}
	\bibliography{references}

\begin{thebibliography}{}

\bibitem[\protect\citeauthoryear{Acemoglu, Ozdaglar, and
  Tahbaz-Salehi}{Acemoglu et~al.}{2016}]{acemoglu2015networks}
Acemoglu, D., A.~Ozdaglar, and A.~Tahbaz-Salehi (2016).
\newblock Networks, shocks, and systemic risk.
\newblock In Y.~Bramoull{\'e}, A.~Galeotti, and B.~Rogers (Eds.), {\em The
  Oxford Handbook of the Economics of Networks}, pp.\  569--608. Oxford
  University Press.

\bibitem[\protect\citeauthoryear{Allen and Gale}{Allen and
  Gale}{2000}]{allen2000financial}
Allen, F. and D.~Gale (2000).
\newblock Financial contagion.
\newblock {\em Journal of Political Economy\/}~{\em 108}, 1--33.

\bibitem[\protect\citeauthoryear{Allers and Elhorst}{Allers and
  Elhorst}{2005}]{Allers2005}
Allers, M. and J.~P. Elhorst (2005).
\newblock Tax mimicking and yardstick competition among local governments in
  the {N}etherlands.
\newblock {\em International Tax and Public Finance\/}~{\em 12}, 493--513.

\bibitem[\protect\citeauthoryear{Anatolyev and Sølvsten}{Anatolyev and
  Sølvsten}{2023}]{Anatolyev2023}
Anatolyev, S. and M.~Sølvsten (2023).
\newblock Testing many restrictions under heteroskedasticity.
\newblock {\em Journal of Econometrics\/}~{\em 236}, 105473.

\bibitem[\protect\citeauthoryear{Arraiz, Drukker, Kelejian, and Prucha}{Arraiz
  et~al.}{2010}]{Arraiz2010}
Arraiz, I., D.~M. Drukker, H.~H. Kelejian, and I.~R. Prucha (2010).
\newblock {A spatial {Cliff-Ord}-type model with heteroskedastic innovations:
  {S}mall and large sample results}.
\newblock {\em Journal of Regional Science\/}~{\em 50}, 592--614.

\bibitem[\protect\citeauthoryear{Ballester, Calv{\'o}-Armengol, and
  Zenou}{Ballester et~al.}{2006}]{ballester2006s}
Ballester, C., A.~Calv{\'o}-Armengol, and Y.~Zenou (2006).
\newblock {Who's who in networks. Wanted: The key player}.
\newblock {\em Econometrica\/}~{\em 74}, 1403--1417.

\bibitem[\protect\citeauthoryear{Blume, Brock, Durlauf, and Ioannides}{Blume
  et~al.}{2011}]{Blume2011}
Blume, L.~E., W.~A. Brock, S.~N. Durlauf, and Y.~M. Ioannides (2011).
\newblock {Identification of Social Interactions}.
\newblock In J.~Benhabib, A.~Bisin, and M.~O. Jackson (Eds.), {\em Handbook of
  Social Economics}, Volume~1, pp.\  853--964. North-Holland.

\bibitem[\protect\citeauthoryear{Bramoull{\'e} and Kranton}{Bramoull{\'e} and
  Kranton}{2007}]{bramoulle2007public}
Bramoull{\'e}, Y. and R.~Kranton (2007).
\newblock Public goods in networks.
\newblock {\em Journal of Economic Theory\/}~{\em 135}, 478--494.

\bibitem[\protect\citeauthoryear{Bramoull{\'e}, Kranton, and
  D'Amours}{Bramoull{\'e} et~al.}{2014}]{Bramoulle2014}
Bramoull{\'e}, Y., R.~Kranton, and M.~D'Amours (2014).
\newblock Strategic interaction and networks.
\newblock {\em American Economic Review\/}~{\em 104}, 898--930.

\bibitem[\protect\citeauthoryear{Calv{\'o}-Armengol, Patacchini, and
  Zenou}{Calv{\'o}-Armengol et~al.}{2009}]{calvo2009peer}
Calv{\'o}-Armengol, A., E.~Patacchini, and Y.~Zenou (2009).
\newblock Peer effects and social networks in education.
\newblock {\em The Review of Economic Studies\/}~{\em 76}, 1239--1267.

\bibitem[\protect\citeauthoryear{Chen, Song, and Yu}{Chen
  et~al.}{2025}]{Chen2025}
Chen, S., X.~Song, and J.~Yu (2025).
\newblock Oracally efficient estimation and specification testing of partially
  linear additive spatial autoregressive models.
\newblock {\em Econometric Reviews\/}~{\em 44}, 1120--1143.

\bibitem[\protect\citeauthoryear{Chen}{Chen}{2007}]{Chen2007}
Chen, X. (2007).
\newblock {\em Large sample sieve estimation of semi-nonparametric models, In:
  Handbook of Econometrics}, Volume~6B, Chapter~76, pp.\  5549--5632.
\newblock North Holland.

\bibitem[\protect\citeauthoryear{Cliff and Ord}{Cliff and
  Ord}{1973}]{cliff1973spatial}
Cliff, A.~D. and J.~K. Ord (1973).
\newblock {\em Spatial Autocorrelation}.
\newblock London: Pion.

\bibitem[\protect\citeauthoryear{Cliff and Ord}{Cliff and
  Ord}{1968}]{Cliff1968}
Cliff, A.~D. and K.~J. Ord (1968).
\newblock The problem of spatial autocorrelation.
\newblock {\em Joint Discussion Paper, University of Bristol: Department of
  Economics, Department of Geography, Series A\/}~{\em 15}, i--88.

\bibitem[\protect\citeauthoryear{de~Jong and Bierens}{de~Jong and
  Bierens}{1994}]{DeJong1994}
de~Jong, R.~M. and H.~J. Bierens (1994).
\newblock On the limit behavior of a {C}hi-square type test if the number of
  conditional moments tested approaches infinity.
\newblock {\em Econometric Theory\/}~{\em 10}, 70--90.

\bibitem[\protect\citeauthoryear{de~Paula}{de~Paula}{2017}]{paula_2017}
de~Paula, {\'A}. (2017).
\newblock {Econometrics of Network Models}.
\newblock In B.~Honor{\'e}, A.~Pakes, M.~Piazzesi, and L.~Samuelson (Eds.),
  {\em Advances in Economics and Econometrics: Eleventh World Congress},
  Volume~1 of {\em Econometric Society Monographs}, pp.\  268–--323.
  Cambridge University Press.

\bibitem[\protect\citeauthoryear{Devereux, Lockwood, and Redoano}{Devereux
  et~al.}{2007}]{Devereux2007}
Devereux, M.~P., B.~Lockwood, and M.~Redoano (2007).
\newblock {Horizontal and vertical indirect tax competition: {T}heory and some
  evidence from the {USA}}.
\newblock {\em Journal of Public Economics\/}~{\em 91}, 451--479.

\bibitem[\protect\citeauthoryear{Devereux, Lockwood, and Redoano}{Devereux
  et~al.}{2008}]{Devereux2008}
Devereux, M.~P., B.~Lockwood, and M.~Redoano (2008).
\newblock {Do countries compete over corporate tax rates?}
\newblock {\em Journal of Public Economics\/}~{\em 92}, 1210--1235.

\bibitem[\protect\citeauthoryear{Donald, Imbens, and Newey}{Donald
  et~al.}{2003}]{Donald2003}
Donald, S.~G., G.~W. Imbens, and W.~K. Newey (2003).
\newblock {Empirical likelihood estimation and consistent tests with
  conditional moment restrictions}.
\newblock {\em Journal of Econometrics\/}~{\em 117}, 55--93.

\bibitem[\protect\citeauthoryear{Gupta}{Gupta}{2018}]{Gupta2018c}
Gupta, A. (2018).
\newblock {Nonparametric specification testing via the trinity of tests}.
\newblock {\em Journal of Econometrics\/}~{\em 203}, 169--185.

\bibitem[\protect\citeauthoryear{Gupta and Qu}{Gupta and Qu}{2024}]{Gupta2022}
Gupta, A. and X.~Qu (2024).
\newblock Consistent specification testing under spatial dependence.
\newblock {\em Econometric Theory\/}~{\em 40}, 278--319.

\bibitem[\protect\citeauthoryear{Gupta and Seo}{Gupta and
  Seo}{2023}]{Gupta2023}
Gupta, A. and M.~H. Seo (2023).
\newblock Robust inference on infinite and growing dimensional time-series
  regression.
\newblock {\em Econometrica\/}~{\em 91}, 1333--1361.

\bibitem[\protect\citeauthoryear{Hong and White}{Hong and
  White}{1995}]{Hong1995}
Hong, Y. and H.~White (1995).
\newblock Consistent specification testing via nonparametric series regression.
\newblock {\em Econometrica\/}~{\em 63}, 1133--1159.

\bibitem[\protect\citeauthoryear{Hoshino}{Hoshino}{2022}]{hoshino2022sieve}
Hoshino, T. (2022).
\newblock Sieve {IV} estimation of cross-sectional interaction models with
  nonparametric endogenous effect.
\newblock {\em Journal of Econometrics\/}~{\em 229}, 263--275.

\bibitem[\protect\citeauthoryear{Jenish}{Jenish}{2016}]{Jenish2016}
Jenish, N. (2016).
\newblock Spatial semiparametric model with endogenous regressors.
\newblock {\em Econometric Theory\/}~{\em 32}, 714--739.

\bibitem[\protect\citeauthoryear{Jenish and Prucha}{Jenish and
  Prucha}{2009}]{Jenish2009}
Jenish, N. and I.~R. Prucha (2009).
\newblock {Central limit theorems and uniform laws of large numbers for arrays
  of random fields}.
\newblock {\em Journal of Econometrics\/}~{\em 150}, 86--98.

\bibitem[\protect\citeauthoryear{Jenish and Prucha}{Jenish and
  Prucha}{2012}]{Jenish2012}
Jenish, N. and I.~R. Prucha (2012).
\newblock On spatial processes and asymptotic inference under near-epoch
  dependence.
\newblock {\em Journal of Econometrics\/}~{\em 170}, 178 -- 190.

\bibitem[\protect\citeauthoryear{Kanbur and Keen}{Kanbur and
  Keen}{1993}]{Kanbur1993}
Kanbur, R. and M.~Keen (1993).
\newblock Jeux sans fronti\`eres: {T}ax competition and tax coordination when
  countries differ in size.
\newblock {\em American Economic Review\/}~{\em 83}, 877--892.

\bibitem[\protect\citeauthoryear{Kelejian and Prucha}{Kelejian and
  Prucha}{1998}]{kelejian1998generalized}
Kelejian, H.~H. and I.~R. Prucha (1998).
\newblock A generalized spatial two-stage least squares procedure for
  estimating a spatial autoregressive model with autoregressive disturbances.
\newblock {\em Journal of Real Estate Finance and Economics\/}~{\em 17},
  99--121.

\bibitem[\protect\citeauthoryear{Kelejian and Prucha}{Kelejian and
  Prucha}{2010}]{kelejian2010specification}
Kelejian, H.~H. and I.~R. Prucha (2010).
\newblock Specification and estimation of spatial autoregressive models with
  autoregressive and heteroskedastic disturbances.
\newblock {\em Journal of Econometrics\/}~{\em 157}, 53--67.

\bibitem[\protect\citeauthoryear{K{\"o}nig, Rohner, Thoenig, and
  Zilibotti}{K{\"o}nig et~al.}{2017}]{konig2017networks}
K{\"o}nig, M.~D., D.~Rohner, M.~Thoenig, and F.~Zilibotti (2017).
\newblock Networks in conflict: {T}heory and evidence from the {G}reat {W}ar of
  {A}frica.
\newblock {\em Econometrica\/}~{\em 85}, 1093--1132.

\bibitem[\protect\citeauthoryear{Korolev}{Korolev}{2026}]{Korolev2026}
Korolev, I. (2026).
\newblock A consistent heteroskedasticity-robust {LM}-type specification test
  for semiparametric models.
\newblock {\em Journal of Applied Econometrics\/}~{\em 41}, 240--252.

\bibitem[\protect\citeauthoryear{Kuersteiner}{Kuersteiner}{2019}]{Kuersteiner2019}
Kuersteiner, G.~M. (2019).
\newblock {Invariance principles for dependent processes indexed by Besov
  classes with an application to a Hausman test for linearity}.
\newblock {\em Journal of Econometrics\/}~{\em 211}, 243--261.
\newblock Annals Issue in Honor of Jerry A. Hausman.

\bibitem[\protect\citeauthoryear{Lee, Phillips, and Rossi}{Lee
  et~al.}{2025}]{Lee2024}
Lee, J., P.~C.~B. Phillips, and F.~Rossi (2025).
\newblock Heteroskedasticity robust specification testing in spatial
  autoregression.
\newblock {\em Econometric Theory\/}~{\em 41}, 995--1043.

\bibitem[\protect\citeauthoryear{Lee and Robinson}{Lee and
  Robinson}{2016}]{Lee2016}
Lee, J. and P.~M. Robinson (2016).
\newblock {Series estimation under cross-sectional dependence}.
\newblock {\em Journal of Econometrics\/}~{\em 190}, 1--17.

\bibitem[\protect\citeauthoryear{Lee}{Lee}{2001}]{Lee2001}
Lee, L.~F. (2001).
\newblock Generalized method of moments estimation of spatial autoregressive
  processes.
\newblock {\em Unpublished manuscript, The Ohio State University\/}.

\bibitem[\protect\citeauthoryear{Lee}{Lee}{2002}]{lee2002consistency}
Lee, L.~F. (2002).
\newblock Consistency and efficiency of least squares estimation for mixed
  regressive, spatial autoregressive models.
\newblock {\em Econometric Theory\/}~{\em 18}, 252--277.

\bibitem[\protect\citeauthoryear{Lee}{Lee}{2004}]{lee2004asymptotic}
Lee, L.~F. (2004).
\newblock Asymptotic distributions of quasi-maximum likelihood estimators for
  spatial autoregressive models.
\newblock {\em Econometrica\/}~{\em 72}, 1899--1925.

\bibitem[\protect\citeauthoryear{Lyytik\"ainen}{Lyytik\"ainen}{2012}]{Lyytikaeinen2012}
Lyytik\"ainen, T. (2012).
\newblock Tax competition among local governments: {E}vidence from a property
  tax reform in {F}inland.
\newblock {\em Journal of Public Economics\/}~{\em 96}, 584--595.

\bibitem[\protect\citeauthoryear{Malikov and Sun}{Malikov and
  Sun}{2017}]{Malikov2017}
Malikov, E. and Y.~Sun (2017).
\newblock {Semiparametric estimation and testing of smooth coefficient spatial
  autoregressive models}.
\newblock {\em Journal of Econometrics\/}~{\em 199}, 12--34.

\bibitem[\protect\citeauthoryear{Pinkse and Slade}{Pinkse and
  Slade}{2010}]{Pinkse2010}
Pinkse, J. and M.~E. Slade (2010).
\newblock The future of spatial econometrics.
\newblock {\em Journal of Regional Science\/}~{\em 50}, 103--117.

\bibitem[\protect\citeauthoryear{Ramsey}{Ramsey}{1969}]{Ramsey1969}
Ramsey, J.~B. (1969).
\newblock Tests for specification errors in classical linear least squares
  regression analysis.
\newblock {\em Journal of the Royal Statistical Society, Series B.\/}~{\em 31},
  350--371.

\bibitem[\protect\citeauthoryear{Redoano}{Redoano}{2014}]{Redoano2014}
Redoano, M. (2014).
\newblock {Tax competition among {E}uropean countries. {D}oes the {EU} matter?}
\newblock {\em European Journal of Political Economy\/}~{\em 34}, 353--371.

\bibitem[\protect\citeauthoryear{Robinson}{Robinson}{2008}]{Robinson2008a}
Robinson, P.~M. (2008).
\newblock Correlation testing in time series, spatial and cross-sectional data.
\newblock {\em Journal of Econometrics\/}~{\em 147}, 5--16.

\bibitem[\protect\citeauthoryear{Su and Qu}{Su and Qu}{2017}]{Su2017}
Su, L. and X.~Qu (2017).
\newblock {Specification test for spatial autoregressive models}.
\newblock {\em Journal of Business {\&} Economic Statistics\/}~{\em 35},
  572--584.

\bibitem[\protect\citeauthoryear{Tincani}{Tincani}{2026}]{Tincani2018}
Tincani, M. (2026).
\newblock Peer effects and rank concerns in the classroom.
\newblock \textit{Journal of Political Economy}, forthcoming.

\bibitem[\protect\citeauthoryear{Tullock}{Tullock}{1980}]{Tullock1980}
Tullock, G. (1980).
\newblock Efficient rent seeking.
\newblock In J.~Buchanan, R.~Tollison, and G.~Tullock (Eds.), {\em Toward a
  Theory of the Rent-Seeking Society}. Texas A{\&}M University Press.

\bibitem[\protect\citeauthoryear{Yang, Song, and Yu}{Yang
  et~al.}{2024}]{Yang2024}
Yang, Z., X.~Song, and J.~Yu (2024).
\newblock Model checking in partially linear spatial autoregressive models.
\newblock {\em Journal of Business \& Economic Statistics\/}~{\em 42},
  1210--1222.

\bibitem[\protect\citeauthoryear{Yang, Song, and Yu}{Yang
  et~al.}{2025}]{Yang2025}
Yang, Z., X.~Song, and J.~Yu (2025).
\newblock Estimation of spatial autoregressive panel data models with
  nonparametric endogenous effect.
\newblock {\em Journal of Econometrics\/}~{\em 252}, 106112.

\end{thebibliography}
\end{document}